\documentclass{article}
\pdfoutput=1

\usepackage{microtype}%
\usepackage[T1]{fontenc}
\usepackage{lmodern}
\usepackage[utf8]{inputenc}
\usepackage{xspace}
\usepackage[autostyle]{csquotes}
\usepackage{amsmath,amsfonts,amssymb, amsthm, mathtools}
\usepackage{fullpage}
\usepackage[table]{xcolor}
\definecolor{darkblue}{rgb}{0,0,0.45}
\definecolor{darkred}{rgb}{0.6,0,0}
\definecolor{darkgreen}{rgb}{0.13,0.5,0}

\usepackage{hyperref}
\hypersetup{colorlinks, linkcolor=darkblue, citecolor=darkgreen,urlcolor=darkblue}
\usepackage[capitalize]{cleveref}
\usepackage{authblk}
\usepackage{tabularx}

\newtheorem{theorem}{Theorem}[section]

\newtheorem{lemma}[theorem]{Lemma}

\newtheorem{claim}[theorem]{Claim}
\newtheorem{proposition}[theorem]{Proposition}

\usepackage{tikz}
\usetikzlibrary{positioning,calc,shapes,fit}
\usetikzlibrary{arrows.meta}
\usetikzlibrary{matrix} %

\title{Parameterized Complexity of Fair Vertex Evaluation Problems\thanks{The conference version has appeared at MFCS 2019~\cite{KMTMFCS}.}}

\date{}
\author[1]{Dušan Knop\thanks{e-mail: \texttt{dusan.knop@fit.cvut.cz} and \href{https://orcid.org/0000-0003-2588-5709}{orcid 0000-0003-2588-5709}.}}
\affil[1]{Department of Theoretical Computer Science, Faculty of Information Technology,\\ Czech Technical University in Prague, Prague, Czech Republic}%

\author[2,3]{Tomáš Masařík\thanks{ e-mail: \texttt{masarik@mimuw.edu.pl} and \href{http://orcid.org/0000-0001-8524-4036}{orcid 0000-0001-8524-4036}. The author was supported by SVV--2017--260452 of Charles University and by the CE-ITI grant project P202/12/G061 of GA ČR.}}
\affil[2]{ %
University of Warsaw, 
Poland
}%

\author[3]{Tomáš Toufar}
\affil[3]{%
Charles University, Prague, Czech Republic}%

\def\phi{\varphi}

\newcommand{\nd}{\ensuremath{\mathop{\mathrm{nd}}}}
\newcommand{\vc}{\ensuremath{\mathop{\mathrm{vc}}}}
\newcommand{\tc}{\ensuremath{\mathop{\mathrm{tc}}}}
\newcommand{\tw}{\ensuremath{\mathop{\mathrm{tw}}}}
\newcommand{\cw}{\ensuremath{\mathop{\mathrm{cw}}}}
\newcommand{\cvd}{\ensuremath{\mathop{\mathrm{cvd}}}}
\newcommand{\td}{\ensuremath{\mathop{\mathrm{td}}}}
\newcommand{\sd}{\ensuremath{\mathop{\mathrm{sd}}}}
\newcommand{\fvs}{\ensuremath{\mathop{\mathrm{fvs}}}}
\newcommand{\mw}{\ensuremath{\mathop{\mathrm{mw}}}}
\newcommand{\Tab}{\mathrm{Tab}}
\newcommand{\FPT}{\ensuremath{\mathsf{FPT}}\xspace}
\newcommand{\XP}{\ensuremath{\mathsf{XP}}\xspace}
\newcommand{\W}[1]{\ensuremath{\mathsf{W[#1]}}}
\newcommand{\NP}{{\ensuremath{\mathsf{NP}}}\xspace}

\newcommand{\Ll}{{\ensuremath{\mathsf{L}}}\xspace}

\newcommand{\FO}{{$\mathsf{FO}$}\xspace}
\newcommand{\MSO}{{$\mathsf{MSO}$}\xspace}
\newcommand{\MSOt}{{$\mathsf{MSO}_2$}\xspace}
\newcommand{\MSOo}{{$\mathsf{MSO}_1$}\xspace}
\newcommand{\bigO}[1]{\ensuremath{\mathcal{O}(#1)}}

\newcommand{\setof}[2]{\left\{#1 : #2 \right\}}
\newcommand{\Setof}[2]{\bigl\{#1 : #2 \bigr\}}

\let\phi=\varphi
\DeclareMathOperator{\df}{\coloneqq}
\DeclareMathOperator{\poly}{{\rm poly}}

\if false
\AtEndEnvironment{proof}{
\unskip\unskip~\hfill\qed
}
\fi

\def\checkmark{\tikz\fill[scale=0.4](0,.35) -- (.25,0) -- (1,.7) -- (.25,.15) -- cycle;}
\newcommand{\easy}[0]{\ea \checkmark}
\newcommand{\hard}[0]{\ha \checkmark}
\newcommand{\unkn}[0]{?}
\newcommand{\ea}[0]{\cellcolor{green!25}}
\newcommand{\ha}[0]{\cellcolor{red!50}}
\newcommand{\eau}[0]{\ea $\ast$}

\newcommand{\empt}[0]{\ensuremath{\mathsf{L_\emptyset}}}

\newcommand{\prob}[3]{%
\begin{center}
\bgroup
\renewcommand{\arraystretch}{0.85}%
\begin{tabularx}{\textwidth}{|lX|}
	\hline
	\multicolumn{2}{|l|}{\hspace{0pt}\rule{0pt}{13pt}#1}\\
	{\bf Input:\enspace}&{#2}\\
	{\bf Question:\enspace}&{#3\rule[-6pt]{0pt}{6pt}} \\
	\hline
\end{tabularx}
\egroup
\end{center}
}

\newcommand{\fixedProb}[4]{%
\begin{center}
\bgroup
\renewcommand{\arraystretch}{0.85}%
\begin{tabularx}{\textwidth}{|lX|}
	\hline
	\multicolumn{2}{|l|}{\hspace{0pt}\rule{0pt}{13pt}#1}\\
	{\bf Input:\enspace}&{#2}\\
	{\bf Fixed:\enspace}&{#3}\\
	{\bf Question:\enspace}&{#4\rule[-6pt]{0pt}{1pt}}\\
	\hline
\end{tabularx}
\egroup
\end{center}
}

\if10     %
\usepackage[mathlines]{lineno}
\newcommand*\patchAmsMathEnvironmentForLineno[1]{%
  \expandafter\let\csname old#1\expandafter\endcsname\csname #1\endcsname
  \expandafter\let\csname oldend#1\expandafter\endcsname\csname end#1\endcsname
  \renewenvironment{#1}%
     {\linenomath\csname old#1\endcsname}%
     {\csname oldend#1\endcsname\endlinenomath}}%
\newcommand*\patchBothAmsMathEnvironmentsForLineno[1]{%
  \patchAmsMathEnvironmentForLineno{#1}%
  \patchAmsMathEnvironmentForLineno{#1*}}%
\AtBeginDocument{%
\patchBothAmsMathEnvironmentsForLineno{equation}%
\patchBothAmsMathEnvironmentsForLineno{align}%
\patchBothAmsMathEnvironmentsForLineno{flalign}%
\patchBothAmsMathEnvironmentsForLineno{alignat}%
\patchBothAmsMathEnvironmentsForLineno{gather}%
\patchBothAmsMathEnvironmentsForLineno{multline}%
}
\linenumbers
\fi

\newtheorem{observation}[theorem]{Observation}

\let\paragraph\subparagraph

\begin{document}

\maketitle

\begin{abstract}
Graph modification problems form a unifying framework for many classical graph problems, where the goal is to remove a minimum number of vertices so that the resulting graph satisfies a desired property~$\pi$.
In this work, we replace the classical objective of minimizing the solution size with the \emph{fair} objective [Lin\&Sahni:\ Fair edge deletion problems.\ IEEE Trans.\ Comput.\ 89], which minimizes the maximum number of deleted neighbors incident to any vertex.
This fairness criterion aims to distribute deletions evenly across the graph and leads to new algorithmic and complexity challenges.

We study \emph{fair deletion and evaluation problems} in the logical setting, where the target property~$\pi$ is defined by a formula in a logic~\textsf{L}.
In particular, we consider the \textsc{Fair $\mathsf{L}$ Vertex Evaluation} problem, which asks whether there exists a vertex set~$W$ of bounded fair cost such that a given formula~$\varphi(X)$ holds in the input graph~$G$.

Our main result is an \FPT algorithm for the \textsc{Fair \MSOo Vertex Evaluation} problem with one free variable, parameterized by the twin cover number of the input graph and the quantifier length of the formula.
We further show that extending this result to formulas with~$\ell$ free variables is unlikely: the \textsc{$\ell$-Fair \FO Vertex Evaluation} problem is \W{1}-hard for parameter~$\ell$ even on graphs of twin cover one.
As a concrete case study, we investigate the \textsc{Fair Vertex Cover} problem, one of the simplest fairness-aware graph problems.
We prove that it is \W{1}-hard when parameterized by the combined treedepth and feedback vertex set, but admits an \FPT algorithm for the parameter modular width.
\bigskip

\centerline{
}
\end{abstract}

\section{Introduction}\label{sec:intro}
Many classical graph problems can be formulated as \emph{vertex deletion problems}~\cite{Yannakakis81}.
For a given graph property~$\pi$, the goal is to find the smallest vertex set~$W$ such that the resulting graph $G\setminus W$ satisfies~$\pi$.
This framework naturally captures a wide range of well-studied problems, such as \textsc{Minimum Vertex Cover} (where $\pi$ describes an edgeless graph) and \textsc{Minimum Feedback Vertex Set} (where $\pi$ corresponds to a forest).
The computational complexity of these problems depends on the expressive power of the property~$\pi$, and many such problems are known to be \NP-complete~\cite{Yannakakis78,Watanabe,KriDeo}.

The connection between graph properties and logic provides a unifying framework to study entire families of graph problems.
If a property~$\pi$ can be expressed in a certain logic (e.g., first-order (\FO) or monadic second-order (\MSO)), then we can use logical model checking to determine whether $G\models\pi$ (i.e., $\pi$ holds in the graph~$G$).
A prominent result in this area is Courcelle’s Theorem~\cite{cour}, which states that every property definable in \MSO logic can be decided in linear time on graphs of bounded treewidth.
Specifically, for an $n$-vertex graph of treewidth~$k$ and an \MSO sentence~$\varphi$, one can decide whether $G\models\varphi$ in $f(k,|\varphi|)n$ time, where $f$ is a computable function.
Such algorithms are fixed-parameter tractable (\FPT) for the combined parameter $k+|\varphi|$.
For comprehensive background on parameterized complexity and related algorithm design, we refer to~\cite{CyganFKLMPPS15,df13}.

Model checking has been extensively studied for various logical fragments and graph classes, including modal logic~\cite{Pilipczuk11}, graphs of bounded neighborhood diversity~\cite{lam}, and graphs of bounded shrubdepth~\cite{GanianHNOM19}.
\FO model checking was also shown to be tractable on graphs with bounded degree~\cite{Seese96}, nowhere dense graphs~\cite{GroheKS17}, and some dense graph classes~\cite{GajarskyHOLR16}.
Although kernelization is one of the most powerful general techniques for \FPT algorithms~\cite{Kernelization}, it is rarely used explicitly in model checking.
The few known examples include Lampis’ kernel for \MSO model checking on graphs of bounded neighborhood diversity~\cite{lam} and the subsequent extension to graphs of bounded shrubdepth by Gajarský and Hliněný~\cite{gajarskyH15}.
This line of research also inspired kernel-based extensions to handle so-called \emph{fair objectives}~\cite{tamc,KKMT}.
Our work continues this line of inquiry.

Fair deletion problems, introduced by Lin and Sahni~\cite{LiSah}, modify the objective function of classical deletion problems: instead of minimizing $|W|$, we minimize the maximum number of deleted neighbors per vertex.
Formally, for a graph~$G=(V,E)$ and a property~$\pi$, the task in \textsc{Fair Vertex Deletion} is to find a set~$W\subseteq V$ minimizing $\max_{v\in V}|N(v)\cap W|$ such that $G\setminus W\models\pi$.
This ensures that deletions are distributed ``fairly'' across the graph.
We study the logical generalization of this framework where~$\pi$ is defined by a formula in a logic~\textsf{L}.
The corresponding problem, \textsc{Fair $\mathsf{L}$ Vertex Deletion}, asks whether there exists a vertex set~$W$ of fair cost at most~$\tau$ such that $G\setminus W\models\varphi$.
A natural variant, \textsc{Fair $\mathsf{L}$ Vertex Evaluation}, replaces $\varphi$ with a formula $\varphi(X)$ having a free variable representing the deleted set, enabling richer conditions on~$W$.
For example, if $\varphi(X)$ expresses that $X$ is a dominating set, then we ask for a dominating set~$W$ such that every vertex has at most~$\tau$ neighbors in~$W$; this is not a deletion problem, as the property concerns $W$ itself rather than $G\setminus W$.
The formal definition is as follows.

\prob{\textsc{Fair $\mathsf{L}$ Vertex Evaluation\footnote{This problem is called \textsc{Generalized Fair $\mathsf{L}$ Vertex-Deletion} in~\cite{tamc}. However, we believe that evaluation is a more suitable expression and coincides with standard terminology in logic.}}}
{An~undirected graph $G=(V,E)$ and a positive integer~$\tau$, an~\textsf{L} formula $\varphi(X)$ with one free vertex-set variable.}
{Is there a~set $W \subseteq V$ of fair cost at most $\tau$ such that $G \models \varphi(W)$?}{}

Analogous definitions exist for edge deletions, leading to the \textsc{Fair $\mathsf{L}$ Edge Deletion} problem.
Here, the fair cost of an edge set $F\subseteq E$ is $\max_{v\in V}|\{e\in F : v\in e\}|$, that is, the maximum number of deleted edges incident to a single vertex.
The fair cost objective appears naturally in applications such as \emph{defective coloring}~\cite{defcol,defectiPractical1,defectiPractical2,Belmonte}, where one seeks colorings minimizing the number of adjacent vertices sharing a color.
This can be seen as a fair edge deletion problem ensuring bounded degree in the subgraph induced by each color class.

Once a model checking algorithm is known, there are several natural directions for further research:
\begin{enumerate}
  \item extending the graph classes or logics~\cite{kontienN11},
  \item improving algorithmic complexity or lower bounds~\cite{lam,frickG04,Lampis:13}, or 
  \item optimizing over logical formulas with free variables~\cite{ALS:91,CourcelleMosbah,GO:13}.
\end{enumerate}
Our contribution follows the latter approach, extending model checking to incorporate the fair objective function.

\paragraph{Our Contributions.}
We investigate the parameterized complexity of fair deletion and evaluation problems, with particular focus on the interplay between logical expressiveness and structural graph parameters.
Our main results establish both algorithmic meta-theorems and tight lower bounds, thereby outlining the precise tractability frontier of \textsc{Fair} problems.

\medskip
\noindent\textbf{(1) A meta-theorem for model checking with fair objectives.}
We first show that fairness can be efficiently incorporated into logical model checking when the input graph has bounded twin cover.

\begin{theorem}\label{thm:fairtc}
  The~\textsc{Fair \MSOo Vertex Evaluation} problem parameterized by the twin cover number and the quantifier length of the formula admits an~\FPT algorithm.
\end{theorem}

Algorithmic meta-results for fair objectives were previously known for fair edge deletion problems on graphs of bounded treewidth~\cite{Kolman09onfair}, and for parameterizations by the vertex cover number~\cite{tamc} and by neighborhood diversity~\cite{tamc,KKMT}.
\Cref{thm:fairtc} extends this line of research to the twin cover number.
Twin cover generalizes vertex cover, and it is incomparable with neighborhood diversity: a disjoint union of cliques of pairwise distinct sizes has twin cover number zero but unbounded neighborhood diversity, whereas the complete bipartite graph $K_{n,n}$ has neighborhood diversity two but twin cover number~$n$.
Thus, \Cref{thm:fairtc} yields tractability of fair problems on dense graph classes not covered by the previous results.

\medskip
\noindent\textbf{(2) Structural parameterization of \textsc{Fair Vertex Cover}.}
We then turn to one of the most fundamental special cases -- \textsc{Fair Vertex Cover} -- and give a detailed complexity classification with respect to several graph parameters.

\prob{\textsc{Fair Vertex Cover (Fair VC)}}
{An~undirected graph $G=(V,E)$, and a~positive integer~$\tau$.}
{Is there a~set $W \subseteq V$ of fair cost at most $\tau$ such that $G \setminus W$ is an~independent set?}{}

\begin{theorem}\label{thm:fairVCisWHwrtTD+FVS}
  The~\textsc{Fair VC} problem parameterized by treedepth $\td(G)$ and feedback vertex set $\fvs(G)$ combined is $\W{1}$-hard.
  Moreover, if there exists an~algorithm running in time $f(w)n^{o(w/\log w)}$, where $w\df \td(G)+\fvs(G)$, then the Exponential Time Hypothesis fails.
\end{theorem}

\begin{theorem}\label{thm:fairVC}
The \textsc{Fair VC} problem parameterized by the modular width $\mw(G)$ admits an~\FPT algorithm with running time $2^{\mw(G)}\mw(G) n^3$, where $n$ is the number of vertices in $G$.
\end{theorem}

\begin{lemma}\label{lem:FairVCpolykernelMW}
The \textsc{Fair VC} problem parameterized by the modular width of the input graph does not admit a polynomial kernel, unless $\NP\subseteq\mathsf{coNP/}\poly$.
\end{lemma}

Together, these results delineate the precise boundary between tractable and intractable cases of \textsc{Fair Vertex Cover}, showing that combining treedepth and feedback vertex number parameters leads to hardness, while modular-width parameterization allows efficient algorithms.
We refer the reader to Figure~\ref{fig:classesVC} for an overview of results for the \textsc{Fair VC} problem.
\begin{figure}[bt]
  \begin{minipage}[c]{0.5\textwidth}
	\centering
	\includegraphics{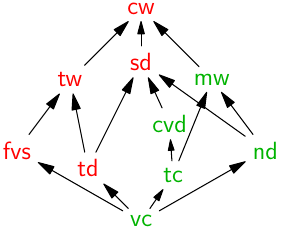}
  \end{minipage}\hfill
  \begin{minipage}[c]{0.5\textwidth}
\caption{Hierarchy of graph parameters for the \textsc{Fair VC} problem. An arrow indicates that a graph parameter upper-bounds the other. Thus, hardness results are implied in the direction of arrows, and \FPT algorithms are implied in the reverse direction.
Green colors indicate \FPT results, and red indicates hardness results.
We denote treewidth by \tw, shrubdepth by \sd{}, and clique width by \cw.
We refer to the book~\cite{CyganFKLMPPS15} for definitions.
Other parameters and their respective abbreviations are defined in Subsection~\ref{sub:prelim}.
The cluster vertex deletion case has since been shown to be \FPT~\cite{cvd}.
}
\label{fig:classesVC}
\end{minipage}
\end{figure}

\medskip
\noindent\textbf{(3) Hardness of general fair evaluation problems.}
Finally, we complement our algorithmic results by establishing strong lower bounds for broader classes of fair evaluation and deletion problems.

\prob{\textsc{$\ell$-Fair $\mathsf{L}$ Vertex Evaluation}}
{An~undirected graph $G=(V,E)$ and a positive integer~$\tau$, an~\textsf{L} formula $\varphi(X_1,\ldots,X_\ell)$ with $\ell$ free vertex-set variables.}
{Are there sets $W_1,\ldots,W_\ell \subseteq V$, each of fair cost at most $\tau$, such that $G \models \varphi(W_1,\ldots,W_\ell)$?}{}

\begin{theorem}\label{thm:lFairVEisWH}
The \textsc{$\ell$-Fair \FO Vertex Evaluation} problem is \W{1}-hard for parameter $\ell$ even when the input graph has twin cover number one.
Moreover, if there exists an algorithm running in time $f(\ell)n^{o(\ell / \log\ell)}$, then the Exponential Time Hypothesis fails.
\end{theorem}

\begin{theorem}\label{thm:cvdEdgeEval}
The \textsc{Fair \FO Edge Deletion} problem is \W{1}-hard when parameterized by the cluster vertex deletion number of the input graph.
\end{theorem}

These lower bounds demonstrate that fairness significantly increases the computational complexity even on highly restricted graph classes.

\medskip
\noindent\textbf{(4) An \XP algorithm for multiple free variables.}
The hardness in Theorem~\ref{thm:lFairVEisWH} is complemented by an \XP algorithm when the number of free variables is included in the parameter.
We prove the following result in Section~\ref{sec:lFairXP}.

\begin{theorem}\label{thm:lFairXP}
The \textsc{$\ell$-Fair \MSOo Vertex Evaluation} problem is in \XP parameterized by the twin cover number, $\ell$, and the quantifier length of the formula.
The polynomial exponent depends only on the twin cover number and $\ell$.
\end{theorem}

\medskip

Overall, our results deepen the understanding of fairness-based objectives in logic-driven graph optimization, bridging the fields of model checking, kernelization, and fair optimization under a unified theoretical framework.
For an overview of the results, please refer to Table~\ref{tab:results} and to Figure~\ref{fig:classes} for the hierarchy of classes.

\begin{table}[bt]
\begin{center}
  \setlength{\tabcolsep}{4.5pt}
  \begin{tabular}{ | l|| c |  c | c | c | c | c |  }
    \hline
	  ~ & $\mathsf{vc}$ & $\mathsf{fvs+td}$& $\mathsf{tc}$ & $\mathsf{nd}$ & $\mathsf{cvd}$ & $\mathsf{mw}$    \\
    \hline\hline
	  \textsc{Fair VC} & \eau & \ha T\ref{thm:fairVCisWHwrtTD+FVS}  & \easy & \eau & \ea \cite{cvd}  &\ea T\ref{thm:fairVC}   \\
	  \textsc{FV $\mathsf{L}$ Del} & $\mathsf{MSO}_2$ \eau & \empt~\hard  &\MSOo \easy & \MSOo \eau  & \FO \ha \cite{cvd}  &\unkn   \\
	  \textsc{FV $\mathsf{L}$ Eval} & $\mathsf{MSO}_2$ \ea \cite{tamc} & \empt~\hard  &\ea \MSOo T\ref{thm:fairtc} & \MSOo \eau  & \FO \ha \cite{cvd}  &\unkn    \\
	  \textsc{$\ell$-FV $\mathsf{L}$ Eval} & $\mathsf{MSO}_1$ \eau  &\empt~\hard &\ha \MSOo T\ref{thm:lFairVEisWH} & \ea \MSOo \cite{KKMT} & \MSOo \hard & \MSOo \hard  \\\hline
    \textsc{FE $\mathsf{L}$ Del} &  $\mathsf{MSO}_2$ \ea \cite{tamc}   &\FO \cite{tamc} \ha  &\unkn & \unkn  & \FO \ha T\ref{thm:cvdEdgeEval} &\unkn   \\
   \hline
  \end{tabular}
\end{center}
\caption{\label{tab:results}
The table summarizes some related (with a citation) and all the presented (with a reference) results on the studied parameters.
Green cells denote \FPT results, and red cells represent hardness results.
Logic $\mathsf{L}$ in metatheorems is specified by a logic used in the respective theorem.
For the $\mathsf{cvd}$ entries for vertex deletion and evaluation, the formula size is included in the parameter.
The twin cover case with $\ell$ free variables also admits an \XP algorithm parameterized by $\tc(G)$, $\ell$, and the quantifier length of the formula (Theorem~\ref{thm:lFairXP}).
By~\empt{} we denote any logic that can express an edgeless graph.
The symbol $\ast$ denotes implied results from previous research, and the symbol \protect\checkmark denotes new implied results.
A question mark (?) indicates that the complexity is unknown.
The \textsc{Fair $\mathsf{L}$ Edge Deletion} problem is delimited from Vertex problems, since there are no apparent relations between them.
}
\end{table}

\begin{figure}[bt]
  \begin{minipage}[c]{0.37\textwidth}
	\centering
	\includegraphics{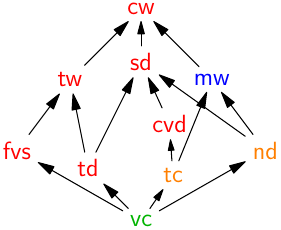}
  \end{minipage}\hfill
  \begin{minipage}[c]{0.63\textwidth}
\caption{Hierarchy of graph parameters with depicted complexity of the \textsc{Fair $\mathsf{L}$ Vertex Evaluation} problem.
An arrow indicates that a graph parameter upper-bounds the other.
Thus, hardness results are implied in the direction of arrows, and \FPT algorithms are implied in the reverse direction.
Green colors indicate \FPT results for \MSOt, orange are \FPT for \MSOo, blue are open, and red are hardness results for \FO.
The cluster vertex deletion case has since been shown to be \W{1}-hard already for \FO, with the formula size included in the parameter~\cite{cvd}.
}
\label{fig:classes}
\end{minipage}
\end{figure}

\subsection{Preliminaries}\label{sub:prelim}
We denote the set $\{1,\ldots,n\}$ by $[n]$.
For a statement $P$, we write $[P]=1$ if $P$ is true and $[P]=0$ otherwise.
We deal with simple undirected graphs; for further standard notation, we refer to monographs on graph theory~\cite{kapitoly} and parameterized complexity~\cite{CyganFKLMPPS15}.
For a vertex $v$ by $N(v)$ we denote the neighborhood of $v$ and by $N[v]$ we denote the closed neighborhood of vertex $v$, i.e., $N(v)\cup \{v\}$.

A set $K\subseteq V$ is a \emph{vertex cover} of $G$ if every edge has an endpoint in $K$; the \emph{vertex cover number} $\vc(G)$ is the minimum size of a vertex cover.
\emph{Twin cover} (introduced by Ganian~\cite{Ganian15}) is a generalization of vertex cover, where we only need to cover edges that are not twin-edges; an edge $\{u,v\}$ is a \emph{twin-edge} if $N(u)\setminus\{v\} = N(v)\setminus\{u\}$.
That is, $K\subseteq V$ is a \emph{twin cover} of $G$ if every edge of $G$ that is not a twin-edge has an endpoint in~$K$; the \emph{twin cover number} $\tc(G)$ is the minimum size of a twin cover.
Note that if $K$ is a twin cover of $G$, then every connected component of $G - K$ is a clique whose vertices all have the same neighborhood in~$K$.
A parameter closely related to twin cover is the \emph{cluster vertex deletion} number (denoted $\cvd(G)$), that is, the smallest number of vertices one has to delete from a graph in order to get a collection of (disjoint) cliques.
\emph{Treedepth} of a graph $G$ (denoted $\td(G)$) is the minimum height of a rooted forest whose transitive closure contains the graph~$G$~\cite{sparsity}.
\emph{Feedback vertex set} (denoted $\fvs(G)$) is the minimum number of vertices of a graph $G$ whose removal leaves a graph without cycles.
\emph{Neighborhood diversity} (denoted $\nd(G)$) is the smallest integer $d$ such that the graph can be partitioned into $d$ sets where each set is either a complete graph or an independent set and each pair of sets either forms a complete bipartite graph or has no edge between them.
\emph{Modular width} of a graph $G$ (denoted $\mw(G)$) is the smallest positive integer $d$ such that $G$ can be obtained from an algebraic expression of width at most $d$, defined as follows~\cite{GLO13}.
The \emph{width of an expression} $\mathcal{E}$ is the maximum number of operands used by any occurrence of the substitution operation in $\mathcal{E}$,
where $\mathcal{E}$ is an algebraic expression that uses the following operations:
\begin{enumerate}
  \item Create an isolated vertex.
  \item The {\em substitution operation} with respect to a template graph $T$ with vertex set $[d]$ and graphs ${G_1,\dots,G_d}$ created by algebraic expressions. The substitution operation, denoted by ${T(G_1, \dots, G_d)},$ results in the graph on vertex set ${V = V_1\cup\dots\cup V_d}$ and edge set
\(
E = E_1\cup\dots\cup E_d\cup \bigcup_{\{i,j\}\in E(T)} \Setof{\{u,v\}}{ u\in V_i, v\in V_j}
,\)
where ${G_i = (V_i,E_i)}$ for all ${i \in [d].}$
\end{enumerate}
An algebraic expression of width $\mw(G)$ can be computed in linear time~\cite{TCHP08,TCHP08-a}.

We follow the standard definitions and notation in logic.
In \FO formulas, free vertex-set variables are understood as unary predicate parameters; quantification is over vertices only.
We consider \emph{vertex-labeled} graphs in some proofs, i.e., graphs with several labels at vertices.
There, the labels refer to assignments of free variables, which is a standard approach in the field~\cite{lam}.
For a comprehensive summary, please consult a book by Libkin~\cite{LibkinFMT}.

\section{Metatheorems for Fair Evaluation}\label{sec:fairMetatheorem}
First, we show an \FPT algorithm for the \textsc{fair \MSOo vertex evaluation} problem parameterized by the twin cover number and the quantifier length of the formula (Theorem~\ref{thm:fairtc}).
We give a more detailed statement that implies the promised result straightforwardly.

We split the proof into two parts.
First, we show an algorithm for \MSOo model checking parameterized by the twin cover number of the graph and the quantifier length of the formula (Proposition~\ref{prop:mctc}).
It is an alternative proof of the result by Ganian~\cite{iwpec/Ganian11}, but this version is useful to be extended for the fair cost.
In the second part, we prove that we can compute the optimal fair cost (Proposition~\ref{prop:tcfair}) and so derive the promised result.

\paragraph*{Overview of the Algorithm.}
For the model checking algorithm, we closely follow the approach of Lampis~\cite{Lampis:13}.
The idea is that if there is a large set of vertices with the same closed neighborhood, then some of them are irrelevant, i.e., we can delete them without affecting the truthfulness of the given formula $\varphi$.
For graphs of bounded neighborhood diversity, using this rule alone can reduce the number of vertices below a bound that depends solely on $\nd(G)$ and $|\varphi|$ only, thus providing an FPT model checking algorithm.
For the graphs of bounded twin cover, this approach can be used to reduce the size of all (twin) cliques, yet their number can still be large.
We take this approach one step further and describe the deletion of irrelevant cliques in a similar manner; these two rules together yield a model checking algorithm for graphs of bounded twin cover.

The reduction rules also lead to a notion of a \emph{shape} of a set $W \subseteq V$.
The motivation behind shapes is to partition all subsets of $V$ with a fixed intersection with the twin cover such that if two sets $W,W'$ have the same shape, then $G \models \varphi(W)$ if and only if $G \models \varphi(W')$.
This allows us to consider only one set of each shape for the purposes of model checking.
Since the number of all distinct shapes is bounded by some function of parameters, we can brute force through all possible shapes as those depend on $\tc(G)$ and the size of the formula.

A final ingredient is an algorithm that, for a given shape, outputs a subset of vertices with this shape that minimizes the fair cost.
This algorithm uses minimum cost flows~\cite{digraphs-book}.

\paragraph{Notation.}
In what follows, $G = (V,E)$ is a graph and $K$ is its twin cover of size $k$.
An \MSOo formula~$\varphi$ contains $q_S$ set quantifiers and $q_v$ element (vertex) quantifiers.
Its quantifier length is $q_S+q_v$.
We use adjacency, equality of element variables, and set membership as atomic predicates.
Other notation, including set equality, is treated as shorthand and is expanded before quantifiers are counted.
By cliques in $G-K$ we mean its nonempty connected components.
Given a twin cover $K$ and $A \subseteq K$, we say that $A$ is the \emph{cover set} of a set $S \subseteq V\setminus K$ if every $v \in S$ has $N(v) \cap K = A$.
Note that, by the definition of twin cover, for all $u,v \in V \setminus K$ with $\{u,v\} \in E$ we have that $A$ is a cover set for $u$ if and only if $A$ is a cover set for $v$.
We say that two cliques in $G- K$ have the same \emph{type} if they have the same size and the same cover set.
Clearly, if the cover set is fixed, two cliques agree on type if and only if their sizes are the same.
We define a \emph{labeled graph}, that is, a graph and a collection of labels on the vertices.
Assignments to free set variables are represented by unary labels, and assignments to free element variables by singleton labels, called \emph{marks}; isomorphisms preserve all labels and marks.
We say that two cliques have the \emph{same labeled type} if all of them have the same size, the same cover set, and the same labels on vertices (in other words, the respective labeled graphs are isomorphic).

\subsection{Model checking}\label{sec:tcModelChecking}
We start with a reformulated combination of Lemma~5 and Theorem~4 by Lampis~\cite{lam}.

\begin{proposition}[{\cite[Lemma 5 and Theorem~4]{lam}}]\label{prop:lam}
Let $\phi$ be an \MSOo formula and let $G$ be a labeled graph.
If there is a set $S$ of more than $2^{q_S}q_v$ twin vertices with the same labels, none of which carries a mark, then for any $v \in S$ we have $G \models \varphi$ if and only if $G \setminus v \models \varphi$.

In particular, if $G$ has one set label and no marked vertices, then for any clique $C$ where each vertex has exactly the same closed neighborhood in $G$, the following holds.
Either there is a vertex $v\in C$ such that $G\models \phi$ if and only if $G\setminus v \models \phi$ or the size of $C$ is bounded by
\[2^{q_S + 1}q_v.\]
\end{proposition}

Proposition~\ref{prop:lam} bounds the size of a maximum clique in $G\setminus K$ because we can apply it repeatedly for each clique that is bigger than the threshold $2^{q_S + 1}q_v$.
Now, we need to bound the number of cliques of each type.
For this, we establish the following technical lemma.

\begin{lemma}\label{lem:induction}
Let $G$ be a labeled graph with twin cover $K$.
Let $\varphi$ be an \MSOo formula with $q_v$ element quantifiers and $q_S$ set quantifiers.
Suppose the size of a maximum clique in $G \setminus K$ is bounded by $\xi$.
If there are strictly more than
\[\alpha_\xi(q_S,q_v)=2^{\xi{q_S}}(q_v+1)\]
cliques of the same labeled type $\mathcal{T}$, then there exists a clique $C$ of type $\mathcal{T}$ such that $G\models \varphi$ if and only if $G\setminus C \models \varphi$.
\end{lemma}

\begin{proof}
For fixed $\xi$, we prove the stronger statement that any prescribed clique $C_0$ of type $\mathcal{T}$ can be deleted.
We use structural induction on $\varphi$, writing universal quantifiers as negated existential quantifiers, which does not change the numbers of set and element quantifiers.
Put $H=G\setminus C_0$, and let $m$ be the number of cliques of type $\mathcal{T}$.
Cliques of the same labeled type can be interchanged by an automorphism fixing $K$ and preserving all existing labels and marks.
Since $m>\alpha_\xi(q_S,q_v)\geq 1$, none of these cliques contains a marked vertex.
Consequently, deleting $C_0$ preserves every atomic formula.
The Boolean cases follow from induction and the monotonicity of $\alpha_\xi$ in both arguments.

Suppose that $\varphi=\exists X\,\psi(X)$, and put $\beta=\alpha_\xi(q_S-1,q_v)$, so that $\alpha_\xi(q_S,q_v)=2^\xi\beta$.
A witness for $X$ in $G$ refines $\mathcal{T}$ into at most $2^\xi$ labeled types.
As $m>2^\xi\beta$, one refined type contains more than $\beta$ cliques.
By induction, deleting a clique $C_1$ of that refined type preserves $\psi$ under this assignment.
An isomorphism from $G\setminus C_1$ to $H$ preserving the original labels transfers the assignment to a witness in $H$.
Conversely, a witness for $X$ in $H$ refines its $m-1\geq 2^\xi\beta$ cliques of type $\mathcal{T}$ into at most $2^\xi$ types.
One of these types occurs at least $\beta$ times.
Extend the assignment to $C_0$ by copying the membership pattern of a clique of that type.
The resulting refined type in $G$ has more than $\beta$ cliques and contains $C_0$.
Induction therefore permits deleting $C_0$ while preserving $\psi$, which proves the reverse implication.

Finally, suppose that $\varphi=\exists x\,\psi(x)$.
Marking a witness changes the labeled type of at most one clique of $\mathcal{T}$, leaving at least
\[
m-1\geq\alpha_\xi(q_S,q_v)>\alpha_\xi(q_S,q_v-1)
\]
cliques of the unchanged type.
For a witness in $G$, induction permits deleting a clique $C_1$ of that unchanged type; an isomorphism from $G\setminus C_1$ to $H$, preserving the original labels, then transfers the witness to $H$.
For a witness in $H$, mark the same vertex in $G$.
The clique $C_0$ has unchanged type, so induction permits deleting $C_0$ and proves the reverse implication.
\end{proof}

From this, we can derive a model-checking algorithm.

\begin{proposition}[Model checking on graphs of bounded twin cover]\label{prop:mctc}
Let $G$ be a graph with twin cover $K$ of size $k$ and the size of the maximum clique in $G\setminus K$ bounded by $2^{q_S}q_v$ and let $\varphi$ be an \MSOo sentence. Then either
\begin{itemize}
  \item there exists a clique $C$ of $G\setminus K$ such that $G\models \phi$ if and only if $G\setminus C \models \varphi$ or
  \item the size of $G$ is bounded by $k+(q_v+1)q_v^2 2^{k+2q_S+2^{q_S}q_S q_v}.$
\end{itemize}
\end{proposition}

\begin{proof}
There are $k$ vertices in the cover and $2^k\xi$ types of cliques, and each of them (by Lemma~\ref{lem:induction}) is repeated at most $\alpha_\xi(q_S,q_v)=2^{\xi{q_S}}(q_v+1)$.
Otherwise, one clique of that type cannot be distinguished by the formula $\phi$.
The maximal size of the clique is $\xi=2^{q_S}q_v$ from Proposition~\ref{prop:lam}, and this gives us the desired bound.
\[\quad\ \ \!k+2^k\xi^2 \alpha_\xi(q_S,q_v)=
k+2^k\xi^2 2^{\xi{q_S}}(q_v+1)=
k+2^k (2^{q_S}q_v)^2 2^{2^{q_S}q_v q_S}(q_v+1)=
\]
$
\quad =\quad
k+(q_v+1) q_v^2 2^{k+2q_S+2^{q_S}q_v q_S}.
$%
\end{proof}

\subsection{Finding a Fair Solution}\label{sec:tcFindingFairSolution}
In the upcoming proof, we follow the ideas of Masařík and Toufar~\cite{tamc}.
They define, for a given formula $\varphi(X)$, a so-called shape of a set $W \subseteq V$ in $G$.
The idea behind a shape is that in order to do the model checking we have deleted some vertices from $G$ that cannot change the outcome of $\varphi(X)$.
However, we have to derive a solution of minimal cost in the whole graph $G$.
Thus, the shape characterizes a set under which $\varphi(X)$ holds, and we have to be able to find a set $W \subseteq V(G)$ for which $\varphi(W)$ holds and $W$ minimizes the fair cost among sets having this shape.

\paragraph{Shape.}
Let $G = (V,E)$ be a graph, $\varphi(X)$ an \MSOo formula, $K\subseteq V$ a twin cover of $G$, $A\subseteq K$, and let $r = 2^{q_S+2}q_v$ and $\alpha = 2^{r (q_S+1)}(q_v+1)$.
Let $\mathcal{C}_A$ be the collection of vertex sets of the connected components of $G-K$ with cover set $A$.
We define an \emph{$A$-shape}.
An $A$-shape of size $r$ is a two dimensional table $S_A$ of dimension $(r+2)\times(r+2)$ indexed by $\left\{0,1,\ldots,r + 1 \right\}\times \left\{ 0,1,\ldots,r + 1 \right\}$.
Each entry $S_A(i,j) \in \left\{ 0,\ldots,\alpha + 1 \right\}$.
The interpretation of $S_A(i,j)$ is the minimum of $\alpha + 1$ and of the number of cliques $C\in\mathcal{C}_A$ such that
\[
\min(r+1, |C \cap W|)=i \textrm{ and } \min(r+1, |C\setminus W|)=j.
\]
Finally, the shape of $W$ in $G$ is a collection of $A$-shapes for all $A\subseteq K$.

The thresholds $r$ and $\alpha$ are chosen with some slack.
If we label the vertices of $W$, then by Proposition~\ref{prop:lam} only $\min(|C\cap W|, 2^{q_S}q_v)$ and $\min(|C\setminus W|, 2^{q_S}q_v)$ matter for every clique~$C$, and both are determined by the cell of~$C$ as $r \ge 2^{q_S}q_v$.
After this reduction, every clique has at most $2^{q_S+1}q_v \le r$ vertices (cf.\ the second part of Proposition~\ref{prop:lam}), and so, by Lemma~\ref{lem:induction} applied with $\xi=r$, only the number of cliques of each labeled type up to $\alpha_r(q_S,q_v) \le \alpha$ matters.

A solution for $\mathcal{C}_A$ can be formally described by a function \mbox{$\mathrm{sol}\colon\mathcal{C}_A\to\{0,\ldots,r+1\}^2$}.
The solution $\mathrm{sol}$ is \emph{compatible} if for every $C\in\mathcal{C}_A$ with $\mathrm{sol}(C) = (i,j)$ either $i+j = |C|$ or $|C|\ge r$, $i = r + 1$ (or $j = r + 1$) and $i+j < |C|$.
For an illustration of a compatible assignment, please refer to Figure~\ref{fig:6x6example}.
We say that a compatible solution $\mathrm{sol}$ \emph{agrees} with shape $S_A$ if $S_A(i,j) = \left| \mathrm{sol}^{-1}(i,j) \right|$, whenever $S_A(i,j) \leq \alpha$ and $\left| \mathrm{sol}^{-1}(i,j) \right| \geq \alpha + 1 $ if $S_A(i,j) = \alpha + 1$.
The $A$-shape $S_A$ is said to be \emph{valid} if there exists a solution that agrees with $S_A$.
Note that such a solution does not exist if the shape specifies too many (or too few) cliques of certain sizes.
The shape $S$ is \emph{valid} if all its $A$-shapes are valid.

The following lemma is a key observation about shapes.
\begin{lemma}\label{lem:shape}
Let $\varphi$ be an \MSOo formula with one free vertex-set variable, $G$ a graph with twin cover $K$, and $W, W'$ two subsets of vertices having the same shape and $W\cap K = W'\cap K$.
Then $G \models \varphi(W)$ if and only if $G \models \varphi(W')$.
\end{lemma}
\begin{proof}
The proof follows using Proposition~\ref{prop:lam} and Lemma~\ref{lem:induction}.
Indeed, if we take the graph $G$ with one label corresponding to the set $W$ and apply the reduction rules given by Proposition~\ref{prop:lam} and Lemma~\ref{lem:induction} and repeat the same process with $W'$, we obtain two isomorphic labeled graphs (see the discussion after the definition of the shape); the isomorphism is the identity on~$K$ since $W\cap K = W'\cap K$.
Note that the same argument applies if $W'$ is a subset of vertices of a different graph $G'$ with $G'[K] = G[K]$ in which $K$ is a twin cover (and the shape of $W'$ is taken in~$G'$).
\end{proof}

With $W\cap K$ fixed, Lemma~\ref{lem:shape} allows us to say that a formula with one free variable holds under a shape since it is irrelevant which subset of vertices of this particular shape is assigned to the free variable.
Also note that deciding whether the formula holds under a realizable shape and a fixed $W\cap K$ can be done in \FPT time parameterized by $k+q_S+q_v$, by picking an assignment respecting them (for example, one given by \cref{lem:computeX}) and running a model-checking algorithm.

\begin{figure}[bt]
 \begin{minipage}[t]{0.49\textwidth}
  \usetikzlibrary{calc,positioning}

\newcommand{\NDist}{.6cm}

\begin{tikzpicture}[node distance=\NDist, font=\scriptsize]
\tikzstyle{ctverec}=[draw, minimum width=\NDist, minimum height=\NDist]
\tikzstyle{ctverecPlny1}=[ctverec, fill=yellow!30]
\tikzstyle{ctverecPlny2}=[ctverec, fill=orange!80]

\node[ctverec, label={0}, label={180:0}] (1:0) {};
\node[ctverec, label={1}, right of=1:0] (1:1) {};
\node[ctverec, label={2}, right of=1:1] (1:2) {};
\node[ctverecPlny1, label={3}, right of=1:2] (1:3) {};
\node[ctverec, label={4}, right of=1:3] (1:4) {};
\node[ctverec, label={5}, right of=1:4] (1:5) {};
\node[ctverecPlny2, label={$\ge6$}, right of=1:5] (1:6) {};

\node[ctverec, below of=1:0, label={180:1}] (2:0) {};
\node[ctverec, right of=2:0] (2:1) {};
\node[ctverecPlny1, right of=2:1] (2:2) {};
\node[ctverec, right of=2:2] (2:3) {};
\node[ctverec, right of=2:3] (2:4) {};
\node[ctverec, right of=2:4] (2:5) {};
\node[ctverecPlny2, right of=2:5] (2:6) {};

\node[ctverec, below of=2:0, label={180:2}] (3:0) {};
\node[ctverecPlny1, right of=3:0] (3:1) {};
\node[ctverec, right of=3:1] (3:2) {};
\node[ctverec, right of=3:2] (3:3) {};
\node[ctverec, right of=3:3] (3:4) {};
\node[ctverec, right of=3:4] (3:5) {};
\node[ctverecPlny2, right of=3:5] (3:6) {};

\node[ctverecPlny1, below of=3:0, label={180:3}] (4:0) {};
\node[ctverec, right of=4:0] (4:1) {};
\node[ctverec, right of=4:1] (4:2) {};
\node[ctverec, right of=4:2] (4:3) {};
\node[ctverec, right of=4:3] (4:4) {};
\node[ctverecPlny2, right of=4:4] (4:5) {};
\node[ctverec, right of=4:5] (4:6) {};

\node[ctverec, below of=4:0, label={180:4}] (5:0) {};
\node[ctverec, right of=5:0] (5:1) {};
\node[ctverec, right of=5:1] (5:2) {};
\node[ctverec, right of=5:2] (5:3) {};
\node[ctverecPlny2, right of=5:3] (5:4) {};
\node[ctverec, right of=5:4] (5:5) {};
\node[ctverec, right of=5:5] (5:6) {};

\node[ctverec, below of=5:0, label={180:5}] (6:0) {};
\node[ctverec, right of=6:0] (6:1) {};
\node[ctverec, right of=6:1] (6:2) {};
\node[ctverecPlny2, right of=6:2] (6:3) {};
\node[ctverec, right of=6:3] (6:4) {};
\node[ctverec, right of=6:4] (6:5) {};
\node[ctverec, right of=6:5] (6:6) {};

\node[ctverecPlny2, below of=6:0, label={180:$\ge6$}] (7:0) {};
\node[ctverecPlny2, right of=7:0] (7:1) {};
\node[ctverecPlny2, right of=7:1] (7:2) {};
\node[ctverec, right of=7:2] (7:3) {};
\node[ctverec, right of=7:3] (7:4) {};
\node[ctverec, right of=7:4] (7:5) {};
\node[ctverec, right of=7:5] (7:6) {};

\node[rotate=90] at (-1.3,-2) {number of vertices in $W$};
\node at (2,1.1) {number of vertices outside $W$};
\end{tikzpicture}
\begin{center}
 \begin{minipage}[c]{0.9\textwidth}
  \caption{\label{fig:6x6example}
  Example of a $7\times7$ $A$-shape.
  All cliques of size 3 will be assigned to yellow (light gray) fields, while cliques of size 8 will be assigned to orange (darker gray) fields.
  }
\end{minipage}
\end{center}
\end{minipage}
 \begin{minipage}[t]{0.49\textwidth}
  \usetikzlibrary{calc,positioning}

\newcommand{\NDist}{.6cm}

\begin{tikzpicture}[node distance=\NDist, font=\scriptsize]
\tikzstyle{ctverec}=[draw, minimum width=\NDist, minimum height=\NDist]
\tikzstyle{ctverecAlert}=[ctverec, fill=red!60]

\newcommand{\NUM}{6}
\pgfmathsetmacro{\NUMm}{int(\NUM - 1)}
\pgfmathsetmacro{\NUMhalf}{(\NUM / 2)*.7}

\node[ctverec, label=0, label={180:0}] (0:0) {0};
\foreach \x[remember=\x as \lastx (initially 0)] in {1,...,\NUMm} {
  \node[ctverec, right of=0:\lastx, label=\x] (0:\x) {0};
}
\node[ctverec, right of=0:\NUMm, label=$\ge\NUM$] (0:\NUM) {0};

\foreach \x[remember=\x as \lastx (initially 0)] in {1,...,\NUMm} {
  \node[ctverec, below of=\lastx:0,label={180:\x}] (\x:0) {\x};
}

\foreach \i[remember=\i as \lasti (initially 0)] in {1,...,\NUMm} {
  \foreach \j[remember=\j as \lastj (initially 0)] in {1,...,\NUM} {
    \node[ctverec, right of=\i:\lastj] (\i:\j) {\i};
  }
}

\node[ctverecAlert, below of=\NUMm:0,label={180:$\ge\NUM$}] (\NUM:0) {?};
\foreach \x in {1,...,\NUMm} {
  \node[ctverecAlert, below of=\NUMm:\x] (\NUM:\x) {?};
}
\node[ctverec, below of=\NUMm:\NUM] (\NUM:\NUM) {\NUM};

\node[rotate=90] at (-1.3,-\NUMhalf) {number of vertices in $W$};
\node at (\NUMhalf,1.1) {number of vertices outside $W$};
\end{tikzpicture}
\begin{center}
 \begin{minipage}[c]{0.9\textwidth}
  \caption{\label{fig:6x6objectiveFunction}%
  The value in the cell is how much we pay for any compatible clique assigned to this cell.
  The value in the last row depends on the size of the clique we are assigning to those cells.
  This part is where we utilize the main power of the minimum cost flows machinery.
  }
\end{minipage}
\end{center}
\end{minipage}
\end{figure}
Lemma~\ref{lem:computeX} computes a realization of an $A$-shape with minimum cardinality, subject to the prescribed local fair-cost bound.
This was done in the previous version of this article~\cite{KMTMFCS,arxiv-v4} using machinery of linear
programming with objective minimization using a separable convex function~\cite{oertelWW14}.
Currently, we use a simpler approach via minimum cost flows.

\paragraph{Preliminaries on Minimum Cost Flows.}
We will follow the terminology and the statements from a book by Bang-Jensen and Gutin~\cite{digraphs-book} in a slightly simplified setting.
A \emph{network} is a sixtuple \[ N=(V_N,E_N,\lambda,\mu,b,\mathrm{cost}),\] where $(V_N,E_N)$ is a directed multigraph, $\lambda,\mu:E_N\to\mathbb{N}_0$ are \emph{lower} (resp.~\emph{capacity}) bounds with $\lambda(e)\le \mu(e)$, $b:V_N\to\mathbb{Z}$ is a \emph{balance (supply/demand)} vector with $\sum_{v\in V_N}b(v)=0$, and $\mathrm{cost}:E_N\to\mathbb{Z}$ is a \emph{cost} function.
A vector $f\in\mathbb{R}^{E_N}$ is a \emph{flow} if $\lambda(e)\le f(e)\le\mu(e)$ for all $e\in E_N$.
Its \emph{balance} at $v\in V_N$ is
\[
b_f(v) \coloneqq \sum_{\substack{e\in E_N\\e\text{ leaves }v}} f(e) - \sum_{\substack{e\in E_N\\e\text{ enters }v}} f(e).
\]
The flow $f$ is \emph{feasible} if $b_f(v)=b(v)$ for all $v\in V_N$.
The \emph{cost} of $f$ is $\mathrm{cost}(f)=\sum_{e\in E_N} \mathrm{cost}(e)f(e)$.
The \textsc{Minimum Cost Flow} problem asks for a feasible flow with minimum cost.

\begin{theorem}[\cite{Klein67}; current wording by {\cite{digraphs-book}[Theorems 4.10.2 and 4.10.3]}]\label{thm:minFlow}
  Given a network $N=(V_N,E_N,\lambda,\mu,b,\mathrm{cost})$.
  If all lower bounds, capacities, and the balance vector are integers, then there exists an algorithm that finds an \textbf{integer} minimum cost flow or concludes that no feasible flow exists in $N$.
  Moreover, the running time of such an algorithm is
\[
\mathcal{O}\big((|V_N|+1)^3+(|V_N|+1)(|E_N|+1)^2\gamma\mu_{\max}\big),
\] where $\mu_{\max}$ is the maximum of $1$ and the capacities, and $\gamma$ is the maximum of $1$ and the absolute arc costs.
The first term accounts for finding an initial feasible integer flow or detecting infeasibility.
\end{theorem}

\paragraph{Fair cost along the twin cover.}
Fix $W_K \subseteq K$ (the guessed set $W \cap K$) and consider a set $W\subseteq V$ with $W \cap K = W_K$.
Let $C \in \mathcal{C}_A$ and $v \in C$.
Since $K$ is a twin cover, we have $N(v) = A \cup (C \setminus \{v\})$, and therefore
\[
  |N(v) \cap W| = |A \cap W_K| + |C \cap W| - [v \in W].
\]
In particular, the maximum fair cost of a vertex of $C$ equals
\[
  \kappa_A(C,s) \df |A \cap W_K| + s - [s = |C|], \qquad\text{where } s = |C \cap W|,
\]
and thus it depends only on~$s$.
On the other hand, every vertex $a \in K$ is adjacent to all vertices of the cliques in $\mathcal{C}_A$ with $a \in A$ and to no other vertex outside~$K$.
Hence
\begin{equation}\label{eq:coverCost}
  |N(a) \cap W| = |N(a) \cap W_K| + \sum_{A \subseteq K \colon a \in A} \ \sum_{C \in \mathcal{C}_A} |C \cap W| .
\end{equation}
Consequently, the constraints for vertices outside $K$ are local to a single clique, whereas for the vertices of $K$ it is best to minimize $\sum_{C \in \mathcal{C}_A} |C \cap W|$ independently for each $A \subseteq K$.

For a clique $C \in \mathcal{C}_A$ and a cell $(i,j)$ compatible with $C$, we define the number of selected vertices of $C$ as
\[
  s(C,i,j) \df
  \begin{cases}
    i & \text{if } i \le r,\\
    |C| - j & \text{if } i = r+1 \text{ and } j \le r,\\
    r+1 & \text{if } i = j = r+1.
  \end{cases}
\]
In the last case, any $s \in \{r+1, \ldots, |C|-r-1\}$ maps $C$ to the cell $(r+1,r+1)$; we use the smallest one, since both $\kappa_A(C,s)$ and the contribution of $C$ to~\eqref{eq:coverCost} are nondecreasing in~$s$.

\begin{lemma}\label{lem:computeX}
Let $G = (V,E)$ be a graph, $K$ be a twin cover of $G$, $A \subseteq K$, $W_K \subseteq K$, and $\tau \in \mathbb{N}_0$.
There is a polynomial-time algorithm that, given an $A$-shape $S_A$ of size $r$, either finds a set $W_A \subseteq \bigcup \mathcal{C}_A$ such that
\begin{itemize}
  \item[(a)]\hypertarget{cond:computeX:a}{} $W_A$ agrees with $S_A$, that is, the solution $C \mapsto \bigl(\min(r+1,|C\cap W_A|), \min(r+1,|C\setminus W_A|)\bigr)$ agrees with $S_A$,
  \item[(b)]\hypertarget{cond:computeX:b}{} $|N(v) \cap (W_K \cup W_A)| \le \tau$ for every $v \in \bigcup \mathcal{C}_A$, and
  \item[(c)]\hypertarget{cond:computeX:c}{} $|W_A|$ is minimum among all sets satisfying \hyperlink{cond:computeX:a}{(a)} and \hyperlink{cond:computeX:b}{(b)},
\end{itemize}
or correctly reports that no set satisfies \hyperlink{cond:computeX:a}{(a)} and~\hyperlink{cond:computeX:b}{(b)}.
\end{lemma}

\begin{proof}
\textbf{Network construction:~~}
Let $\mathcal{C}_A$ be the family of cliques whose cover set is $A\subseteq K$. Build a directed bipartite network in the following way.
For each $C\in \mathcal{C}_A$, create a vertex $v_C$ with supply $1$ ($b(v_C)=1$).
For each cell $(i,j)$ of a shape table $S_A$ create a vertex $w_{i,j}$ with demand $S_A (i, j)$ ($b(w_{i,j})= - S_A (i, j)$) where $(i,j) \in \{0,\ldots,r+1\}^2$.
Let \[\delta \df |\mathcal{C}_A|-\sum_{i=0}^{r+1}\sum_{j=0}^{r+1} S_A (i, j).\]
If $\delta$ is negative, then the shape is invalid and we reject it.
Add one more \emph{target} vertex $t$ with demand $\delta$ ($b(t)=-\delta$).
Now, for each $C\in \mathcal{C}_A$, we add an arc $v_C w_{i,j}$ for every cell $(i,j)$ that is compatible with $C$ (those are cells where the clique can be mapped; consult \Cref{fig:6x6example} for an illustration) and satisfies $\kappa_A(C, s(C,i,j)) \le \tau$.
Such arcs have capacity $1$ and cost $s(C,i,j)$, that is, the number of vertices of $C$ in the solution, should $C$ be mapped to cell $(i,j)$.
Intuitively, this number represents the clique’s contribution to the fair cost on $A$; see~\Cref{fig:6x6objectiveFunction}.
Then we add arcs from vertices $w_{i,j}$ with demand equal to $\alpha +1$ to vertex $t$.
Those arcs have capacity $\delta$ and cost $0$.
All arcs will have lower bounds $0$.

Now, using~\Cref{thm:minFlow} we obtain a minimum cost integer flow or conclude that no feasible flow exists.
Because of the integrality, each clique $C$ chooses exactly one cell $(i,j)$ such that the arc $v_C w_{i,j}$ is present.
For each cell, there is an exact number of cliques choosing it due to demands, and in the case of cells with an $\alpha + 1$ value, at least $\alpha+1$ cliques are choosing it.
We obtain $W_A$ by selecting, in every clique $C$ that chose the cell $(i,j)$, any $s(C,i,j)$ of its vertices (the vertices of $C$ are twins, so the choice does not matter).
Then $W_A$ satisfies \hyperlink{cond:computeX:a}{(a)} by the above, it satisfies \hyperlink{cond:computeX:b}{(b)} since $\kappa_A(C,s(C,i,j)) \le \tau$ for every used arc, and $|W_A|$ equals the cost of the flow.

Conversely, let $W'_A$ be any set satisfying \hyperlink{cond:computeX:a}{(a)} and \hyperlink{cond:computeX:b}{(b)}.
Map every clique $C \in \mathcal{C}_A$ to the cell $(i,j) = \bigl(\min(r+1,|C\cap W'_A|), \min(r+1,|C\setminus W'_A|)\bigr)$.
We have $s(C,i,j) = |C \cap W'_A|$, except for the cell $(r+1,r+1)$, where $s(C,i,j) = r+1 \le |C\cap W'_A|$.
Since $\kappa_A(C,\cdot)$ is nondecreasing, \hyperlink{cond:computeX:b}{(b)} implies $\kappa_A(C,s(C,i,j)) \le \tau$, and thus the arc $v_C w_{i,j}$ is present.
Hence, this mapping (with the surplus routed to~$t$) is a feasible flow of cost at most $|W'_A|$.
It follows that the minimum cost flow yields a set satisfying \hyperlink{cond:computeX:c}{(c)}, and if no feasible flow exists, then no set satisfies \hyperlink{cond:computeX:a}{(a)} and \hyperlink{cond:computeX:b}{(b)}.
Note that for $A = \emptyset$ the set $W_A$ contributes nothing to the fair cost on $K$, so only local feasibility matters when combining the families.
The network has polynomial size in $G$ and the explicitly given shape table.
After rejecting $\delta<0$, all capacities, absolute balances, and arc costs are at most $|V(G)|$, so \Cref{thm:minFlow} gives the claimed polynomial running time.
\end{proof}

Now we are ready to prove the main result of this section.
It essentially follows by the exhaustive search among all possible shapes $S$ such that $\varphi$ is true under $S$ and application of \Cref{lem:computeX}.

\begin{proposition}\label{prop:tcfair}
Let $G = (V,E)$ be a graph with twin cover $K$ of size $k$.
For an \MSOo formula $\varphi(X)$ with one free vertex-set variable, it is possible to find a set $W\subseteq V$ such that
\begin{itemize}
  \item
  $G\models \varphi(W)$  and
  \item
  the cost of $W$ is minimized among all subsets of $V$ satisfying $\varphi(X)$
\end{itemize}
in time $f(k,q_S+q_v)(|V|+|\varphi|)^{O(1)}$ for some computable function $f$ or report that no such $W$ exists.
\end{proposition}
\begin{proof}
The high-level idea of the algorithm is as follows.
For every possible selection of $K\cap W$, we generate all possible shapes.
We first check whether the formula $\varphi(X)$ evaluates to true under the shape $S$.
If the shape~$S$ is valid, we compute a solution $W \subseteq V$ that agrees with $S$ having the minimal fair-cost.
Note that this is done separately for every set of twin-cliques $\mathcal{C}_A$ for $A \subseteq K$; for this, we use Lemma~\ref{lem:computeX}.
Finally, we return the set $W$ minimizing the fair cost among these (in fact, the union of those computed in the previous step).

\noindent\textbf{Shape evaluation.}
Let $W \cap K$ be fixed and let $S$ be a shape.
We reject any guess with $S_A(0,0)>0$ for some $A$, since components are nonempty.
We create an auxiliary graph~$H$ consisting of $G[K]$ and, for every $A\subseteq K$ and every cell $(i,j)$, exactly $S_A(i,j)$ cliques of size $i+j$ with cover set~$A$.
We label vertices of~$H$ based on the shape~$S$ and the initial selection of $W\cap K$, that is, in each clique created for cell $(i,j)$ we label exactly $i$ vertices, and in $K$ we label precisely $W\cap K$.
By construction, the labeled set has shape~$S$ in~$H$.
For every $W\subseteq V$ with shape~$S$ and the given $W\cap K$, Lemma~\ref{lem:shape} (and the note after its proof) implies that $\varphi$ holds for the labeled set in~$H$ if and only if $G\models\varphi(W)$.
Since $|V(H)|$ is bounded by a function of $k$ and $q_S+q_v$, direct recursive model checking on~$H$ takes time
\[
\mathcal{O}\bigl(|\varphi|(|V(H)|+1)^{q_v}2^{|V(H)|q_S}\bigr).
\]
If $\varphi(X)$ evaluates to true under $S$, we proceed with the shape realization described below.
Thus, this routine runs in \FPT time with respect to $k$ and $q_S+q_v$ as parameters, with polynomial dependence on the full formula length $|\varphi|$.

\noindent\textbf{Shape realization.}
Now, we focus on the case where we are given $W_K = W \cap K$ and a valid shape~$S$ under which the formula~$\varphi$ holds in~$G$, and we have to find a realization of~$S$ in~$G$ minimizing its fair cost.
To this end, we fix a bound $\tau$ on the fair cost and decide whether there is a realization of fair cost at most~$\tau$.
For every $A \subseteq K$, including $A = \emptyset$, we apply \Cref{lem:computeX} to $S_A$, $W_K$, and~$\tau$.
If it reports that no suitable set exists for some $A$, we reject the pair $(W_K,S)$ for this~$\tau$.
Otherwise, we set $W \df W_K \cup \bigcup_{A\subseteq K} W_A$ and accept if $|N(a)\cap W| \le \tau$ for every $a \in K$.

By \hyperlink{cond:computeX:b}{(b)} of \Cref{lem:computeX}, only the vertices of $K$ need to be checked.
Moreover, by~\eqref{eq:coverCost} and \hyperlink{cond:computeX:c}{(c)} of \Cref{lem:computeX}, the set $W$ minimizes $|N(a)\cap W|$ simultaneously for all $a \in K$ among all sets $W'$ with $W' \cap K = W_K$, shape~$S$, and $|N(v)\cap W'| \le \tau$ for all $v \notin K$.
Thus, if there is a set $W^*$ of fair cost at most $\tau$ with $G\models\varphi(W^*)$, then the run with $W_K = W^*\cap K$ and $S$ being the shape of $W^*$ accepts.
Conversely, every accepted set $W$ has fair cost at most $\tau$, and $G \models \varphi(W)$ holds by \Cref{lem:shape}, since $W$ has shape~$S$.
To find the optimum, we run the above for $\tau = 0,1,\ldots,|V|$ and return the set found for the smallest accepted~$\tau$.

In order to summarize the running time, we have
\begin{itemize}
  \item
  at most $|V|+1$ values of $\tau$,
  \item
  at most $2^k$ possible selections for $W\cap K$, and
  \item
  at most $(\alpha+2)^{2^k(r + 2)^2}$ possible shapes (for each such selection), as there are $2^k$ sets $A\subseteq K$ and each $A$-shape has $(r+2)^2$ entries from $\{0,\ldots,\alpha+1\}$.
\end{itemize}
For every such guess, we perform evaluation and (possibly) realization; the evaluation has the bound given above, and the realization consists of $2^k$ calls of \Cref{lem:computeX}, each running in polynomial time.
Since $\alpha$ and $r$ are bounded by functions of $q_S+q_v$, the total running time is $f(k,q_S+q_v)(|V|+|\varphi|)^{O(1)}$ for some computable function $f$.
This finishes the proof.
\end{proof}

A minimum twin cover can be found in \FPT time by applying the standard vertex-cover branching algorithm to the graph obtained from $G$ by deleting all twin-edges, determined in the original graph $G$.
Its vertex covers are exactly the twin covers of $G$, so \Cref{prop:tcfair} implies \Cref{thm:fairtc} even when a twin cover is not supplied.

\section{The \textsc{Fair VC} problem}\label{sec:fairVC}
\subsection{Hardness for Treedepth and Feedback Vertex Set}\label{sec:hardness}
We begin with several simple observations about the fair objective value $\tau$ in the \textsc{Fair VC} problem.
\begin{observation}\label{obs:fairVC-preselectedVertex}
Let $G = (V,E)$ be a graph and $U\subseteq V$ be a vertex cover of $G$ with fair objective $\tau$, that is, $\forall v\in V$ it holds that $|N(v)\cap U|\le \tau$.
If $v\in V$ has $\deg(v) > \tau$, then $v\in U$.
\end{observation}
Note that we can use Observation~\ref{obs:fairVC-preselectedVertex} to force a vertex $v$ to be a part of the fair vertex cover by attaching $\tau+1$ degree one vertices to $v$.
Observe further that we may adjust (lower) the global budget $\tau$ for an individual vertex $v$ to $\rho\in\{0,\ldots,\tau\}$ on its original neighbors by attaching $\tau-\rho$ vertices to $v$ and then attaching $\tau$ new leaves to each newly added vertex.
To this end, if the above operations are applied to a graph $G$ resulting in a graph $G'$, then $\td(G')\le\td(G) + 2$ and $\fvs(G') = \fvs(G)$.

We observe a substantial connection between \textsc{Fair VC} and \textsc{Target Set Selection} (TSS).
It is worth mentioning that \textsc{Vertex Cover} can be formulated in the language of TSS by setting the threshold to $\deg(v)$ for every vertex $v$.
As a result, our reduction given here is, in a certain sense, dual to the one given by Chopin, Nichterlein, Niedermeier, and  Weller~\cite{ChopinNNW14} for the TSS problem.
However, we will show that the structure of the solution for \textsc{Fair VC} is, in fact, the complement of the structure of the solution for TSS given therein.

We reduce from the \emph{Partitioned Subgraph Isomorphism (PSI)} problem, which becomes a standard problem to reduce from, as it gives strong bounds under ETH due to a powerful result by Marx~\cite{MarxPSI}.
\prob{\textsc{Partitioned Subgraph Isomorphism} (PSI)\hspace{14em} {\em Parameter:} $|E(H)|$}
{An $\ell$-partite host graph $G=(V_1 \uplus \cdots \uplus V_\ell,E)$ such that $|V_i|=|V_j|$ for all $i\neq j$ and a pattern graph $H$ with $V(H)=\{1,\dots,\ell\}$ without isolated vertices, where each $V_a$ is an independent set.}
{Is there a mapping (called a \emph{homomorphism})  $h:V(H)\to V_1 \uplus \cdots \uplus V_\ell$ such that $h(a)\in V_a$ for all $a\in[\ell]$ and $\{h(a),h(b)\}\in E$ if $\{a,b\}\in E(H)$?}

\begin{figure}[tb]
  \begin{center}
    \resizebox{\textwidth}{!}{\newcommand{\fairVCNodeDistance}{.65cm}
\usetikzlibrary{positioning,calc,fit}
\usetikzlibrary{decorations.pathreplacing}

\begin{tikzpicture}[node distance=\fairVCNodeDistance]
  \tikzstyle{dummyVrchol}=[inner sep=1pt]
  \tikzstyle{vrchol}=[circle, draw, inner sep=1pt]
  \tikzstyle{emptyVrchol}=[circle, draw, inner sep=3pt]
  \tikzstyle{fitting}=[rounded corners, draw, inner sep=5pt]
  \tikzstyle{hrana}=[thick]
  \tikzstyle{selected}=[fill=gray!80]

  \node[vrchol] (Va1) {1};
  \node[vrchol, below of=Va1] (Va2) {2};
  \node[vrchol, below of=Va2] (Va3) {$i$};
  \node[vrchol] at ($(Va3) - 2*(0,\fairVCNodeDistance)$) (Va4) {$n$};
  \node[dummyVrchol] at ($(Va3)!.5!(Va4.north)$) (VaD) {$\vdots$};
  \node[fitting, fit=(Va1)(Va2)(Va3)(Va4),label={90:$V_a$}] {};

  \node[emptyVrchol, label={[xshift=-3pt]90:guard}, selected, label={[xshift=-4pt]270:$n-1$}] at ($(Va3) - 2*(\fairVCNodeDistance,0)$) (VaGuard) {};

  \foreach \vrchol in {(Va1),(Va2),(Va3),(Va4)} {
     \draw[hrana] (VaGuard) -- \vrchol;
  }

  \node[emptyVrchol] at ($(Va1) + 2*(\fairVCNodeDistance, 0) + 2*(0, \fairVCNodeDistance)$) (VaILow1) {};
  \node[emptyVrchol, below of=VaILow1] (VaILow2) {};
  \node[emptyVrchol, below of=VaILow2] (VaILow3) {};
  \node[fitting, fit=(VaILow1)(VaILow2)(VaILow3), label={90:$i$},label={55:lower}] {};

  \node[emptyVrchol] at ($(VaILow3) - 3*(0,\fairVCNodeDistance)$) (VaIHigh1) {};
  \node[emptyVrchol, below of=VaIHigh1] (VaIHigh2) {};
  \node[emptyVrchol, below of=VaIHigh2] (VaIHigh3) {};
  \node[emptyVrchol] at ($(VaIHigh3) - 2*(0,\fairVCNodeDistance)$) (VaIHigh4) {};
  \node[dummyVrchol] at ($(VaIHigh3)!.5!(VaIHigh4.north)$) {$\vdots$};
  \node[fitting, fit=(VaIHigh1)(VaIHigh2)(VaIHigh3)(VaIHigh4), label={90:$n-i$},label={290:upper}] {};

  \node[vrchol, selected] at ($(VaILow3) + 3*(\fairVCNodeDistance, 0) - (0, \fairVCNodeDistance)$) (Iab1) {$c^1_{ab}$};
  \node[vrchol, selected, below=of Iab1] (Iab2) {$c^2_{ab}$};
  \node[fitting, fit=(Iab1)(Iab2), label={270:$n$}] {};

  \foreach \vrchol in {(VaILow1),(VaILow2),(VaILow3)} {
    \draw[hrana] (Va3) -- \vrchol;
    \draw[hrana] (Iab1) -- \vrchol;
  }
  \foreach \vrchol in {(VaIHigh1),(VaIHigh2),(VaIHigh3),(VaIHigh4)} {
    \draw[hrana] (Va3) -- \vrchol;
    \draw[hrana] (Iab2) -- \vrchol;
  }

  \node[emptyVrchol] at ($(VaILow1) + 6*(\fairVCNodeDistance, 0)$) (IabIHigh1) {};
  \node[emptyVrchol, below of=IabIHigh1] (IabIHigh2) {};
  \node[emptyVrchol, below of=IabIHigh2] (IabIHigh3) {};
  \node[emptyVrchol] at ($(IabIHigh3) - 2*(0,\fairVCNodeDistance)$) (IabIHigh4) {};
  \node[dummyVrchol] at ($(IabIHigh3)!.5!(IabIHigh4.north)$) {$\vdots$};
  \node[fitting, fit=(IabIHigh1)(IabIHigh2)(IabIHigh3)(IabIHigh4), label={90:$n-i$}, label={110:$a$-upper}] {};

  \node[emptyVrchol]  at ($(IabIHigh4) - 3*(0,\fairVCNodeDistance)$) (IabILow1) {};
  \node[emptyVrchol, below of=IabILow1] (IabILow2) {};
  \node[emptyVrchol, below of=IabILow2] (IabILow3) {};
  \node[fitting, fit=(IabILow1)(IabILow2)(IabILow3), label={90:$i$}, label={[yshift=.2cm]230:$a$-lower}] {};

  \node[vrchol] at ($(IabIHigh1) + 5*(\fairVCNodeDistance, 0)$) (Eab1) {1};
  \node[vrchol, below of=Eab1] (Eab2) {2};
  \node[vrchol, below of=Eab2] (Eab3) {3};
  \node[vrchol, below of=Eab3] (Eab4) {4};
  \node[vrchol, below of=Eab4] (Eab5) {$q$};
  \node[vrchol] at ($(Eab5) - 2*(0, \fairVCNodeDistance)$) (Eab6) {$m$};
  \node[dummyVrchol] at ($(Eab5)!.5!(Eab6.north)$) {$\vdots$};
  \node[fitting, fit=(Eab1)(Eab2)(Eab3)(Eab4)(Eab5)(Eab6),label={90:$E_{\{a,b\}}$}] {};

  \foreach \vrchol in {(IabILow1),(IabILow2),(IabILow3)} {
    \draw[hrana] (Iab2) -- \vrchol;
    \draw[hrana] (Eab5) -- \vrchol;
  }
  \foreach \vrchol in {(IabIHigh1),(IabIHigh2),(IabIHigh3),(IabIHigh4)} {
    \draw[hrana] (Iab1) -- \vrchol;
    \draw[hrana] (Eab5) -- \vrchol;
  }

  \node[emptyVrchol, label={90:guard}, label={0:$m-1$}, selected] at ($(Eab2) + 2.2*(\fairVCNodeDistance,0)$) (EabGuard) {};
  \foreach \vrchol in {(Eab1),(Eab2),(Eab3),(Eab4),(Eab5),(Eab6)} {
    \draw[hrana] (EabGuard) -- \vrchol;
  }

  \node[emptyVrchol] at ($(Eab5) + 4.2*(\fairVCNodeDistance, 0) + 1.5*(0,\fairVCNodeDistance)$) (EabJHigh1) {};
  \node[emptyVrchol, below of=EabJHigh1] (EabJHigh2) {};
  \node[emptyVrchol] at ($(EabJHigh2) - 1.6*(0,\fairVCNodeDistance)$) (EabJHigh3) {};
  \node[dummyVrchol] at ($(EabJHigh2)!.5!(EabJHigh3.north)$) {$\vdots$};
  \node[fitting, fit=(EabJHigh1)(EabJHigh2)(EabJHigh3), label={90:$n-j$}, label={[xshift=.15cm]60:$b$-upper}] {};

  \node[emptyVrchol] at ($(EabJHigh3) - 2*(0,\fairVCNodeDistance)$) (EabJLow1) {};
  \node[emptyVrchol, below of=EabJLow1] (EabJLow2) {};
  \node[emptyVrchol] at ($(EabJLow2) - 1.6*(0,\fairVCNodeDistance)$) (EabJLow3) {};
  \node[dummyVrchol] at ($(EabJLow2)!.5!(EabJLow3.north)$) {$\vdots$};
  \node[fitting, fit=(EabJLow1)(EabJLow2)(EabJLow3), label={90:$j$}, label={[xshift=.1cm]290:$b$-lower}] {};

  \node[vrchol, selected] at ($(EabJHigh2) + 2.6*(\fairVCNodeDistance, 0) - (0,.3*\fairVCNodeDistance)$) (Iba1) {$c^1_{ba}$};
  \node[vrchol, selected] at ($(EabJLow2) + 2.6*(\fairVCNodeDistance, 0) - (0,.3*\fairVCNodeDistance)$) (Iba2) {$c^2_{ba}$};
  \node[fitting, fit=(Iba1)(Iba2), label={270:$n$}] {};

  \foreach \vrchol in {(EabJHigh1),(EabJHigh2),(EabJHigh3)} {
    \draw[hrana] (Eab5) -- \vrchol;
    \draw[hrana] (Iba1) -- \vrchol;
  }
  \foreach \vrchol in {(EabJLow1),(EabJLow2),(EabJLow3)} {
    \draw[hrana] (Eab5) -- \vrchol;
    \draw[hrana] (Iba2) -- \vrchol;
  }

  \foreach \dy in {.25cm, 0cm, -.25cm} {
    \draw[hrana, dashed] (Iba1) -- ++(.7cm, \dy);
    \draw[hrana, dashed] (Iba2) -- ++(.7cm, \dy);
  }
  \node[anchor=west] at ($(Iba1) + (.75cm, 0)$) {$V_b$-lower};
  \node[anchor=west] at ($(Iba2) + (.75cm, 0)$) {$V_b$-upper};
\end{tikzpicture}}
  \end{center}
  \caption{An overview of the reduction in the proof of Theorem~\ref{thm:fairVCisWHwrtTD+FVS}, showing the selection gadgets for $V_a$ and $E_{\{a,b\}}$ (only the enumeration vertices of the $i$-th and the $q$-th choice vertex are drawn) and both incidence check gadgets of the pattern edge $\{a,b\}$; the selection gadget for $V_b$ is not drawn. The gray vertices are forced to be a part of the fair vertex cover. If a vertex's fair objective was lowered, then the resulting threshold is beneath the vertex (the group of vertices).}\label{fig:fairVChardnessBig}
\end{figure}

\begin{proof}[Proof of Theorem~\ref{thm:fairVCisWHwrtTD+FVS}]
Let $(G,H)$ be an instance of \textsc{PSI} with color classes $V_1,\ldots,V_\ell$ in $G$ and let $n = \left|V_a\right|$ for all $a\in [\ell]$.
For each edge $\{a,b\}\in E(H)$, let $E_{\{a,b\}}$ denote the set of edges between $V_a$ and $V_b$. 
We may assume that $H$ is nonempty and $E_{\{a,b\}}\neq\emptyset$ for every $\{a,b\}\in E(H)$, since the remaining cases are trivial.
We can also suppose that $|E_{\{a,b\}}|\ge 2n+\ell+1$ and is the same for all $ab\in E(H)$.
Indeed, we can multiply some edges of the original $G$ (temporarily creating a multigraph), so that the number requirement is met while increasing the size only polynomially in $n$ and $\ell$.
As a consequence we set $m \df |E_{\{a,b\}}|$ for arbitrary $ab\in E(H)$.
We may discard all host edges between $V_a$ and $V_b$ whenever $a\neq b$ and $\{a,b\}\notin E(H)$.
We will describe a graph $G' = (D,F)$ that, together with $\tau= m-1$, forms an equivalent instance of the \textsc{Fair VC} problem.
The reduction has the following properties:
\begin{itemize}
  \item $|D| = \poly(n, \tau)$ and $|F| = \poly (m,\tau)$, 
  \item $\td(G') \le 4|E(H)|+3$ and $\fvs(G') \le 4|E(H)|$.
\end{itemize}
For an overview of the reduction, please refer to Figure~\ref{fig:fairVChardnessBig}.
There are three types of gadgets in our reduction, namely the \emph{vertex-selection gadget} (one for each color $a\in[\ell]$), the \emph{edge-selection gadget} (one for each pattern edge $\{a,b\}\in E(H)$), and the \emph{incidence check gadget} (one for each vertex--edge incidence). 
We start by enumerating the vertices in each color class by numbers from $[n]$ and, for each $\{a,b\}\in E(H)$, the edges $E_{\{a,b\}}$ by numbers in $[m]$.
Throughout the proof, $a,b$ are distinct elements of $[\ell]$.

The $V_a$ selection gadget consists of $n$ \emph{choice} vertices (representing the color class $V_a$), a special vertex called a \emph{guard}, and a group of $n^2$ \emph{enumeration} vertices.
The guard vertex is connected to all choice vertices; it is forced to be a part of the fair vertex cover, and its budget is lowered so that at most $n-1$ choice vertices can be in any fair vertex cover.
The $i$-th choice vertex is connected to $n$ enumeration vertices.
For each choice vertex, there are $n$ such vertices, and so these are private for vertex $i$.
We further divide these vertices into two parts—the \emph{lower part} consists of $i$ vertices and the \emph{upper part} consists of $n-i$ vertices.

For each pattern edge $\{a,b\}\in E(H)$, the $E_{\{a,b\}}$ selection gadget consists of $m$ \emph{choice} vertices, a special vertex called a \emph{guard}, and a group of $2n m$ \emph{enumeration} vertices. It is constructed analogously to the vertex selection gadget.
If the $q$-th edge of $E_{\{a,b\}}$ connects the $i$-th vertex in $V_a$ and the $j$-th vertex in $V_b$, then $2n$ (private) enumeration vertices are connected to the $q$-th choice vertex.
These are split into four groups---the \emph{lower $a$-part} consisting of $i$ vertices, the \emph{upper $a$-part} consisting of $n-i$ vertices, and similarly the \emph{lower} and \emph{upper $b$-parts}.

For each ordered pair $(a,b)$ with $\{a,b\}\in E(H)$ (that is, for each of the two incidences of the pattern edge $\{a,b\}$), the $ab$-incidence check gadget consists of two vertices $c^1_{ab}$ and $c^2_{ab}$.
Both $c^1_{ab}$ and $c^2_{ab}$ are forced to be a part of the solution and have a lowered budget so that at most $n$ enumeration neighbors of each of them can be part of any fair vertex cover.
The vertex $c^1_{ab}$ is connected to every lower-part vertex in the selection gadget for $V_a$ and to every upper $a$-part vertex in the selection gadget for $E_{\{a,b\}}$.
The vertex $c^2_{ab}$ is connected to every upper-part vertex in the selection gadget for $V_a$ and to every lower $a$-part vertex in the selection gadget for $E_{\{a,b\}}$.
This finishes the construction of $G'$.

\medskip

Now observe that if we remove all the $4|E(H)|$ vertices $c^1_{ab}, c^2_{ab}$ from $G'$, then each component of the resulting graph is a tree that can be rooted with height at most $3$.
It follows that $\td(G') \le 4|E(H)|+3$ and $\fvs(G') \le 4|E(H)|$, as is the size of the removed set of vertices.
We finish the proof by showing that the two instances are equivalent.

Suppose $(G,H)$ is a yes-instance witnessed by a valid homomorphism $h$ with $h(a)\in V_a$ for all $a$, and denote by $i_a\in[n]$ the index of $h(a)$ in $V_a$.
For each $\{a,b\}\in E(H)$, let $q_{ab}=q_{ba}\in[m]$ be the index of the host edge $\{h(a),h(b)\}\in E_{\{a,b\}}$.
We now construct a vertex cover $C_h$ of $G'$ with $|N(w)\cap C_h|\le \tau$ for all $w\in D$.
The set $C_h$ consists of the following:
\begin{itemize}
  \item all enforced vertices (including all guard and check vertices);
  \item for each $a\in[\ell]$, all $V_a$-choice vertices \emph{except} the $i_a$-th choice vertex, and all enumeration vertices of the $i_a$-th choice vertex (both its lower and upper parts);
  \item for each $\{a,b\}\in E(H)$, all $E_{\{a,b\}}$-choice vertices \emph{except} the $q_{ab}$-th choice vertex, and all enumeration vertices of the $q_{ab}$-th choice vertex.
\end{itemize}
It is a routine check to verify that $C_h$ is a vertex cover and no vertex exceeds the fair cost $\tau$.

Conversely, suppose that $C$ is a vertex cover in $G'$ such that $|N(w)\cap C|\le \tau$ for all $w\in D$.
Recall that all enforced vertices must belong to $C$.
We show that for each $a\in[\ell]$, \emph{exactly one} $V_a$-choice vertex is \emph{not} in $C$, and for each $\{a,b\}\in E(H)$ \emph{exactly one} $E_{\{a,b\}}$-choice vertex is \emph{not} in $C$.
Indeed, if two distinct $V_a$-choice vertices (say $i$ and $j$) were both not in $C$, then all their private enumeration vertices would be in $C$.
Since $H$ has no isolated vertices, there exists $b$ with $\{a,b\}\in E(H)$;
considering $c^1_{ab}$ and $c^2_{ab}$ together, their neighborhoods inside the $V_a$-gadget contribute $2n$ vertices of $C$ from the two choices, and, additionally, the $a$-parts of the $E_{\{a,b\}}$-gadget force at least another $n$ vertices in $C$ (at least one $E_{\{a,b\}}$-choice vertex is not in $C$ due to its guard), giving at least $3n>2n$ enumeration vertices of $C$ in $N(c^1_{ab})\cup N(c^2_{ab})$. This contradicts the lowered budget of $n$ on each of $c^1_{ab}$ and $c^2_{ab}$.
Thus, each $V_a$-gadget contributes exactly one choice vertex that is not in $C$; call its index $i_a$.
The same argument shows exactly one $E_{\{a,b\}}$-choice vertex is not in $C$; call its index $q_{ab}$.
Finally, $c^1_{ab}$ and $c^2_{ab}$ can both have at most $n$ enumeration neighbors in $C$ only if the vertex selected in $V_a$ is incident in $G$ to the edge selected in $E_{\{a,b\}}$, and symmetrically for $c^1_{ba}$, $c^2_{ba}$ and $V_b$; hence the chosen host vertices $h(a)$, for $a\in[\ell]$ (with $h(a)$ being the $i_a$-th vertex of $V_a$), and the host edges indexed by $q_{ab}$ realize a partition-respecting copy of $H$ in $G$.
Therefore, $(G,H)$ is a yes-instance of \textsc{PSI}.

We conclude the argument by proving a lower bound based on ETH.
By a seminal result of Marx~\cite{MarxPSI}, \textsc{Partitioned Subgraph Isomorphism} with parameter $\eta\df |E(H)|$  admits no algorithm running in time $f(\eta)n^{o(\eta/\log \eta)}$ unless ETH fails.
Our reduction creates one vertex-selection gadget per pattern vertex and one edge-selection gadget and two incidence check gadgets per pattern edge.
The parameter changes as follows:
\[
  w\df \td(G')+\fvs(G') = \bigO{\left(|E(H)|\right)} =\bigO{\eta}.
\]
As the size $N$ of the constructed instance $G'$ is polynomial in $n$ and $\ell$, and $\ell\le2\eta$, an $f(w)N^{o(w/\log w)}$-time algorithm for \textsc{Fair VC} would imply a $g(\eta)n^{o(\eta/\log \eta)}$-time algorithm for \textsc{PSI}, for some function $g$, contradicting~\cite{MarxPSI}.
\end{proof}

\subsection{FPT algorithm for Modular Width}

We may assume $V(G)\neq\emptyset$, since otherwise the instance is trivially positive.
Since the algebraic expression $\mathcal{E}$ of width $\mw(G)$ can be computed in linear time~\cite{TCHP08,TCHP08-a}, we can assume that we have $\mathcal{E}$ as input.
We construct the rooted ordered tree $\mathcal{T}$ corresponding to $\mathcal{E}$.
Each node $t \in \mathcal{T}$ is assigned a graph $G^t \subseteq G$, that is, the graph constructed by the subexpression of $\mathcal{E}$ rooted at $t$.
Suppose we are performing a substitution operation at node $t$ with respect to a template graph $T$ and graphs $G_1,\ldots, G_d$.
Here $d=|V(T)|\le\mw(G)$.
Denote the resulting graph by $G^t$ and denote by $n_i$ the size of $V(G_i)$.

The computation will be carried out recursively from the bottom of the tree $\mathcal{T}$.
We first describe the structure of all vertex covers $C$ in $G^t$.
Since there are at most $2^d$ vertex covers of $T$, we try all of these.
Furthermore, every set $C \cap V(G_i)$ must be a vertex cover of $G_i$.
We will describe the fair cost of the cover $C$ in terms of fair costs and sizes of the sets $C \cap V(G_i)$.
Based on some of their properties we are going to derive, we will design a dynamic program that computes the solution.

\begin{proof}[Proof of Theorem~\ref{thm:fairVC}]
The edges in $G^t$ between two vertices of $G_i$ will be referred to as \emph{old edges}; the edges between $G_i$ and $G_j$ for $i \neq j$ (i.e., edges introduced by the template operation) will be referred to as \emph{new edges}.

The computation is carried out recursively from the bottom of the tree~$\mathcal{T}$.

We first describe the structure of all vertex covers $C$ in $G^t$.
Observe that if $ij \in E(T)$ then at least one of $V(G_i), V(G_j)$ must be a subset of $C$. Otherwise, there would be a new edge not covered by $C$. From this, we can see that the set $C_T \df \setof{i}{V(G_i) \subseteq C}$ is a vertex cover of the template graph $T$. We call $C_T$ the \emph{type of the vertex cover $C$}. Furthermore, every set $C \cap V(G_i)$ must be a vertex cover of $G_i$ (otherwise there would be an old edge uncovered by $C$).

We now describe the fair cost of the cover $C$ in terms of fair costs and sizes of the sets $C \cap V(G_i)$. Denote by $c_i$ the size $|C \cap V(G_i)|$ and by $f_i$ the fair cost of $C \cap V(G_i)$ in $G_i$.
The \emph{fair cost of $C$ in $Z \subseteq V(G)$} is defined as $\max_{v\in Z} |C \cap N(v)|$.
For $i \in [d]$ the fair cost of $C$ in $V(G_i)$ can be written as
$ f_i + \sum_{j:ij\in E(T)} c_j. $
Clearly, the fair cost of $C$ is the maximum of the last expression over $i \in [d]$.

If we know the type $C_T$ of the cover $C$, this can be simplified based on whether $i$ lies in $C_T$. If $i \in C_T$ then $f_i$ is $\Delta(G_i)$ (the maximal degree of $G_i$). If, on the other hand, $i \notin C_T$, then all its neighbors $j$ are in $C_T$ and in this case $c_j = n_j$. Combining those observations, we have
\[
	\max_{v\in V(G_i)} |N_{G^t}(v)\cap C| =
	\begin{cases}
		\Delta(G_i) + \sum_{j\notin C_T:ij\in E(T)} c_j + \sum_{j\in C_T:ij\in E(T)} n_j & i  \in C_T, \\
		f_i + \sum_{j:ij\in E(T)} n_j & i \notin C_T. \\
	\end{cases}
\]

For each node $t$ of the tree $\mathcal{T}$ we keep an $|V(G^t)|$-entry table $\Tab^t$, indexed by $p\in\{0,\ldots,|V(G^t)|-1\}$, with values in $\{0,\ldots,n\}\cup\{\infty\}$.
The value at position $\Tab^t[p]$ is the smallest size of a cover in $G^t$ of fair cost at most $p$ or $\infty$ if such a cover does not exist; note that $\Tab^t$ is nonincreasing in~$p$.
For $p\ge |V(G^t)|$, we use $\Tab^t[p]=\Tab^t[|V(G^t)|-1]$, since no cover of $G^t$ has fair cost exceeding $|V(G^t)|-1$.

The computation of $\Tab$ in leaves of $\mathcal{T}$ is trivial: $\Tab^t[0]=0$.
We describe how to compute the value $\Tab^t[p]$ given that we know $\Tab_i$ in all children of $t$.
It is enough to consider vertex covers $C$ of $G^t$ of fixed type; we initialize $\Tab^t[p]=\infty$ and iterate over all types as there are at most $2^d$ of them.

Fix a type $C_T$.
A cover $C$ of type $C_T$ and fair cost at most $p$ exists if and only if there are nonnegative values $p_i$ with $\Tab_i[p_i]\neq\infty$ for every $i\notin C_T$ satisfying~\eqref{eq:first} and~\eqref{eq:second}.
For such values, the update is given in~\eqref{eq:third}:
\begin{align}
	p &\geq \Delta(G_i) + \sum_{j\notin C_T:ij \in E(T)} \Tab_j[p_j] +\sum_{j\in C_T:ij \in E(T)} n_j & \ \forall i \in C_T, & \label{eq:first} \\
	p &\geq p_i + \sum_{j:ij \in E(T)} n_j & \ \forall i \notin C_T, & \label{eq:second} \\
	\Tab^t[p] &\gets \min\left\{\Tab^t[p],\ \sum_{j \notin C_T} \Tab_j[p_j] +\sum_{j \in C_T} n_j\right\}. &  & \label{eq:third}
\end{align}
First, for every $i \notin C_T$ we set $p_i$ to the highest possible value without violating the inequality~\eqref{eq:second}, that is, $p_i := p - \sum_{j:ij \in E(T)} n_j$.
Note that this is always a safe choice; by monotonicity of the tables, increasing $p_i$ can only decrease $\Tab_i[p_i]$, so it cannot violate~\eqref{eq:first} or increase the candidate size.
Clearly, if any such $p_i$ is negative, we know that the given constraints cannot be satisfied and there is no vertex cover $C$ of given type and fair cost.

If for any $i\notin C_T$ $\Tab_i[p_i]=\infty$, then there is no vertex cover $C_i$ in $G_i$ of fair cost at most $p_i$.
This means that we cannot find a cover $C$ of given type and fair cost, so we discard the type $C_T$ (it does not contribute to $\Tab^t[p]$). %
We check whether inequalities~\eqref{eq:first} hold.
If not, we discard the type $C_T$ as well.
Otherwise, we perform the update in~\eqref{eq:third}.
We claim that there is a vertex cover $C$ of given type  and fair cost; we can set
	$C = \bigcup_{i \in C_T} V(G_i) \cup \bigcup_{i \notin C_T} C_i,$
where $C_i$ is any vertex cover of $G_i$ of fair cost at most $p_i$ and size $\Tab_i[p_i]$ (this is guaranteed to exist because $\Tab_i[p_i]$ was not $\infty$).
Each such $C_i$ is a proper subset of $V(G_i)$, since removing any one vertex from a full cover preserves the cover property and cannot increase its fair cost.
It is straightforward to check that $C$ has the required properties.
Moreover, from our choice of values $p_i$ it follows that a vertex cover of type $C_T$ and fair cost at most $p$ exists if and only if the described procedure finds values of $p_i$, and the constructed cover has minimum size among all such covers.
By taking the minimum over all types $C_T$, we can fill the value $\Tab^t[p]$ as required.

To complete the description of the algorithm, it is enough to check whether $\Tab_{\text{root}}[\tau]\neq\infty$, where $\tau$ was the desired fair cost. %

The running time is $O(n)$ for the induction over the expression $\mathcal{E}$ (after suppressing unary substitutions) times $2^d$ different types of covers in any single node, times filling the table of size at most $n$  times $nd$ for determining the correct values $p_i$ and checking other inequalities for every $i\in[d]$. Since $d\le\mw(G)$ at every node, this altogether yields a $2^{\mw(G)}\mw(G)n^3$ time algorithm.
\end{proof}

\begin{proof}[Proof of Lemma~\ref{lem:FairVCpolykernelMW}]
First, we observe that for any two nonempty graphs $G_1, G_2$ it holds that $\mw(G_1 \dot\cup G_2) \le \max\{2,\mw(G_1),\mw(G_2)\}$, where $\dot\cup$ denotes the disjoint union.
Indeed, disjoint union is a substitution with an edgeless two-vertex template.
Now, it remains to show that \textsc{Fair VC} admits an AND-cross-composition from an \NP-hard problem, see~\cite[Chapter~15]{CyganFKLMPPS15}; the rest then follows from the framework of Bodlaender, Downey, Fellows, and Hermelin~\cite{bodlaenderDFH09} together with the result of Drucker~\cite{Drucker15}, which rules out a polynomial kernel in this case unless $\NP\subseteq\mathsf{coNP/}\poly$.
To this end, observe that if a graph $G$ is not connected, then $U$ is a vertex cover in $G$ if and only if $U \cap C$ is a vertex cover in $G[C]$ for every connected component $C$ of $G$.

More precisely, we give an AND-cross-composition from the unparameterized \textsc{Fair VC} problem on nonempty graphs, which is \NP-hard (e.g., by the polynomial-time reduction from \textsc{PSI} in the proof of Theorem~\ref{thm:fairVCisWHwrtTD+FVS}).
Let $(G_i,\tau_i)$ for $i\in[t]$ be such instances and let $n_{\max} = \max_{i\in[t]}|V(G_i)|$.
Every instance with $\tau_i \ge |V(G_i)|$ is a yes-instance (take $U = V(G_i)$), and so we can replace $\tau_i$ by $\min(\tau_i,|V(G_i)|)$.
We use the equivalence relation that groups valid instances by equality of these normalized thresholds; all other strings form one additional class, for which we output a fixed no-instance.
This is a polynomial equivalence relation~\cite[Chapter~15]{CyganFKLMPPS15}, since it is polynomial-time decidable and the normalized thresholds are bounded by the input encoding length.
For inputs in one valid class, we may therefore assume that $\tau_1 = \cdots = \tau_t = \tau$.
We output the instance $(G_1 \dot\cup \cdots \dot\cup G_t, \tau)$.
By the above observation (applied to each component), it is a yes-instance if and only if $(G_i,\tau)$ is a yes-instance for every $i\in[t]$, since the fair cost of $U$ is the maximum of the fair costs of $U\cap V(G_i)$ in $G_i$.
Its modular width is at most $\max\{2,\max_{i\in[t]}\mw(G_i)\}\le\max\{2,n_{\max}\}$, which is polynomial in the size of the largest input instance, as required.
\end{proof}

\section{Hardness of Possible Extensions}\label{sec:hard}
We use the \textsc{Unary $\ell$-Bin Packing} problem as the starting point of our hardness reduction.
\textsc{Unary $\ell$-Bin Packing} is \W{1}-hard for parameter $\ell$, the number of bins to be used~\cite{JansenKMS13}.
Here, the item sizes are encoded in unary, and the task is to assign $n$ items to $\ell$ bins such that the sum of sizes of items assigned to any bin does not exceed its capacity $B$.
Formally, \textsc{Unary $\ell$-Bin Packing} is defined as follows.
\prob{\textsc{Unary $\ell$-Bin Packing}\hspace{21em} {\em Parameter:} $\ell$}
{Positive integers $\ell$ and $B$ and item sizes $s_1, \ldots, s_n$ encoded in unary.}
{Is there a packing of all items into at most $\ell$ bins of size $B$?}

\begin{proof}[Proof of Theorem~\ref{thm:lFairVEisWH}]
We may assume that $\ell<n$ and $s_i\le B$ for every $i\in[n]$, since all remaining cases are decidable in polynomial time.
We construct a formula $\varphi(X_1,\ldots, X_\ell)$ and an instance $(G,\tau)$ of \textsc{$\ell$-Fair \FO Vertex Evaluation} with $\tau = B$ from an instance of \textsc{Unary $\ell$-Bin Packing} as follows.
The graph $G$ is formed by $n$ disjoint cliques and a universal vertex $u$.
Cliques in $G$ represent the items by their respective sizes, that is, there is a clique with $s_i$ vertices for every $i\in [n]$; denote the clique representing item $i$ by $C_i$.
The set $\{u\}$ is a twin cover, and $\tc(G)=1$ since $n\ge2$ and an edge joining $u$ to $C_1$ is not a twin-edge.
This finishes the description of the graph; now we turn our attention to the formula.
Free variables $X_1,\ldots, X_\ell$ are going to represent an assignment of items to bins.
Note that it is possible to recognize the universal vertex $u$ by the following \FO formula, since $u$ is the only vertex of $G$ satisfying it:
\[
\operatorname{univ}(v)	\df	(\forall w\in V)\big((w \neq v) \to (wv \in E)\big).
\]
We fix $u$ for the rest of the description of $\varphi(X_1, \ldots, X_\ell)$; this can easily be done by using the form $(\exists u\in V)(\operatorname{univ}(u)\land\cdots)$.
Let $\psi(v)$ be a formula.
We use $(\forall v \in V\setminus\{u\})(\psi(v))$ as a short form of the expression
\[
(\forall v\in V) \big( (v \neq u) \to \psi(v) \big).
\]
Note that this can be straightforwardly extended for more quantifiers.

In order to represent the bin choice, we need to ensure two conditions.
First, every item is packed; that is, every non-universal vertex must belong to some $X_j$.
Second, every item is fully packed inside (at least) one bin, that is, vertices belonging to the same clique agree on $X_j$ membership.
To do so, we first introduce the following predicates (representing these conditions):
\[
\operatorname{cover}(X_1,\ldots,X_\ell)	\df	(\forall v\in V\setminus\{u\}) \left( \bigvee_{j = 1}^\ell v\in X_j \right)
\]
and
\[
\operatorname{same}(X_1,\ldots,X_\ell)	\df	\left( \forall v,w\in V\setminus\{u\} \right) \left( (vw \in E) \to \bigwedge_{j = 1}^\ell  \left( v\in X_j \Leftrightarrow w\in X_j \right) \right).
\]
The construction of the new instance is finished by letting
\[
\varphi(X_1,\ldots, X_\ell) = (\exists u\in V)\big(\operatorname{univ}(u) \land \operatorname{cover}(X_1,\ldots,X_\ell) \land \operatorname{same}(X_1,\ldots,X_\ell)\big).
\]
It remains to argue that the instances are indeed equivalent.
Note that the bin capacity/fair cost is essentially checked only for $u$ since the fair cost of any vertex cannot exceed its degree.
The degree of any other vertex (not $u$) does not exceed $\max_{i \in [n]} s_i$, and this is always upper-bounded by $B$.

Let $(B, \ell, s_1, \ldots, s_n)$ be a Yes-instance of \textsc{Unary $\ell$-Bin Packing}.
This is witnessed by a mapping $\sigma\colon [n] \to [\ell]$ assigning items to bins.
Now, we put
\[
W_j = \cup_{i \in [n] \colon \sigma(i) = j} C_i.
\]
Observe that $\left| N(u) \cap W_j \right| = \sum_{i \in [n], \sigma(i) = j} s_i \le B$ for every $j \in [\ell]$, since $\sigma$ represented a valid assignment.
Furthermore, $\varphi(W_1, \ldots,W_\ell)$ holds, since every clique vertex is covered and vertices of a clique are always in the same $W_j$.

For the opposite direction assume we have a graph $G = (V,E)$ and that there exist sets $W_1, \ldots, W_\ell \subseteq V$, each of fair cost at most $B$, such that $\varphi(W_1, \ldots, W_\ell)$ holds in $G$.
Recall that $u$ is the universal vertex of $G$.
First, we have that $V = \{u\}\cup W_1 \cup \cdots \cup W_\ell$, since the predicate $\operatorname{cover}(W_1, \ldots, W_\ell)$ holds if and only if every vertex in a clique $C_i$ is in some $W_j$.
Furthermore, since the predicate $\operatorname{same}(W_1, \ldots, W_\ell)$ holds, we have that $C_i \cap W_j$ is either an empty set or $C_i$ for every $i\in [n]$ and $j\in [\ell]$.
This allows us to construct an assignment $\sigma\colon [n] \to [\ell]$ by setting $\sigma(i) = j$, where $j \in [\ell]$ is the smallest number such that $W_j \cap C_i = C_i$ holds.
Now, we have that $\sum_{i \in [n], \sigma(i) = j} s_i \le \left|W_j \setminus \{u\}\right| = |N(u)\cap W_j| \le B$ by the fair objective.
This finishes the proof of equivalence since we have constructed a valid assignment.

Finally, the reduction runs in polynomial time and the number of free vertex-set variables is exactly~$\ell$.
Let $L$ denote the encoded length of the \textsc{Unary $\ell$-Bin Packing} instance. The constructed graph has
\[
n_{\mathrm{out}}=|V(G)|=1+\sum_{i=1}^{n}s_i=O(L).
\]
Since \textsc{Unary $\ell$-Bin Packing} admits no $f(\ell)L^{o(\ell/\log \ell)}$-time algorithm unless ETH fails~\cite{JansenKMS13}, the reduction implies that \textsc{$\ell$-Fair \FO Vertex Evaluation} admits no $f(\ell)n_{\mathrm{out}}^{o(\ell/\log \ell)}$-time algorithm unless ETH fails.
\end{proof}

\begin{proof}[Proof of Theorem~\ref{thm:cvdEdgeEval}]
We may assume that $2\le\ell<n$, $s_i\le B$ for every $i\in[n]$, and $\sum_{i\in[n]}s_i>B$, since all remaining cases are decidable in polynomial time.
Multiplying all item sizes and the bin capacity by three, we may also assume that $B\ge3$.
We construct a sentence $\varphi$, a graph $G$, and $\tau$ forming an instance of \textsc{Fair \FO Edge Deletion} with $\tau = B$ from an instance of \textsc{Unary $\ell$-Bin Packing} as follows.
Each item is represented by a clique on $3B$ vertices; denote the clique associated with the $i$-th item by $C_i$.
In addition, there are $\ell$ vertices $v_1,\ldots,v_\ell$ (representing bins) and $\ell$ guard vertices $g_1, \ldots, g_\ell$ (auxiliary vertices that help to recognize bins).
For each $i\in[n]$, choose $s_i$ vertices in $C_i$ and call these vertices \emph{special}.
For each $j\in [\ell]$ we connect $v_j$ with every special vertex and with $g_j$.
This finishes the description of $G$.

The sentence $\varphi$ is constructed using auxiliary predicates, which we describe first.
In the following formulas, adjacency is interpreted in the graph in which the formula is evaluated.
\[
\operatorname{guard}(v) \df (\exists u\in V) ((uv \in E) \land (\forall w\in V\setminus\{u,v\})(wv \notin E)).
\]
A predicate $\operatorname{guard}(v)$ is used to recognize the guard vertices.
Recall that we set the fair cost to $B$ and thus the solution $F$ deletes at most $B$ edges incident to any vertex.
This shows that vertex $v$ must have been a guard vertex (before the deletion of $F$) in order to fulfill the $\operatorname{guard}(v)$ predicate (after the deletion of $F$), since the degree of every clique vertex in $G$ is at least $3B-1$ which is at least $2B-1$ after deleting $F$, and the degree of every bin vertex is $1+\sum_{i\in[n]} s_i \ge B+2$ (we may assume $\sum_{i\in[n]} s_i > B$, as otherwise all items fit into one bin), which is at least $2$ after deleting $F$.
\[
\operatorname{bin}(v)	\df	(\exists u) ((uv\in E) \land \operatorname{guard}(u)).
\]
\[
\operatorname{item-edge}	\df
(\forall v, w \in V : v \neq w)\Bigl(\neg\operatorname{bin}(v)\land \neg\operatorname{bin}(w) \land \bigl((\exists u\in V) (\neg\operatorname{bin}(u) \land vu\in E \land wu\in E) \bigr) \to (vw \in E)\Bigr).
\]
Suppose that the predicate $\operatorname{item-edge}$ holds in $G \setminus F$.
We claim that $vw \in E(G\setminus F)$ for any two distinct vertices $v,w \in C_i$. Since $|C_i| = 3B$, it follows that $v$ and $w$ have at least $B-2$ common neighbors in $C_i$ after deleting $F$.
Thus, the edge $vw$ cannot be deleted, as $v$ and $w$ have a common (non-bin) neighbor (provided $B\ge 3$).
We define an auxiliary sentence
\[
\psi = \left( \exists v_1, \ldots, v_\ell \in V \right)\left( \bigwedge_{j \in [\ell]}\operatorname{bin}(v_j) \land \bigwedge_{j \neq j' \in [\ell]} v_j \neq v_{j'} \right) \land \operatorname{item-edge}
\]
and, since $\psi$ implicitly assures existence of $\ell$ (different) guards, conclude the following claim.
\begin{claim}\label{clm:psi_characterisation}
Let $F$ be a set of edges in $G$ of fair cost at most $B$.
We have that $G\setminus F \models \psi $ if and only if $F$ contains only edges between bins and special vertices.
\end{claim}

The next predicate we introduce is $\operatorname{notable}(v)$, defined as follows:
\[
\operatorname{notable}(v)\df \neg \operatorname{guard}(v) \land \bigl( (\exists u \in V) (\operatorname{bin}(u) \land  uv\in E)\bigr).
\]
\begin{claim}\label{clm:notable_properties}
Let $F$ be a set of edges in $G$ of fair cost at most $B$ such that $G\setminus F\models\psi$.
\begin{enumerate}
  \item If $G \setminus F \models \operatorname{notable}(v)$, then $v$ is a special vertex.
  \item $G \models \operatorname{notable}(v)$ if and only if $v$ is a special vertex.
  \item If $G \setminus F \not\models \operatorname{notable}(v)$ for a special vertex $v \in C_i$, then $N_{G\setminus F}(v) \subseteq C_i$.
\end{enumerate}
\end{claim}
\begin{proof}[Proof of \Cref{clm:notable_properties}]\renewcommand{\qed}{\hfill$\diamondsuit$}
A vertex $v$ is \emph{notable} in a graph if $\operatorname{notable}(v)$ holds in that graph.
By Claim~\ref{clm:psi_characterisation}, bins and guards remain correctly recognized, and the only non-guard neighbors of bins are special vertices; this proves the first part.
The second part follows immediately from the construction of $G$, since the only vertices attached to the bin vertex $v_j$ are the special vertices and their guard vertex $g_j$ (which is not notable).
On the other hand, the set $F$ may contain edges $vv_j$ for all $j \in [\ell]$.
Clearly, such a former special vertex is not notable.
Finally, if $v$ is special and $G \setminus F \not\models \operatorname{notable}(v)$, then $v$ lost all of its edges to bin vertices $v_1, \ldots, v_\ell$. 
\end{proof}

The goal is to delete an edge set $F$ that describes a valid assignment of items into bins.
We need to ensure two conditions.
First, every special vertex is not a neighbor of some bin vertex $v_j$.
Second, if a notable vertex $v \in C_i$ is not a neighbor of a bin vertex $v_j$, then all notable vertices in $C_i$ are not neighbors of $v_j$.
The two conditions correspond to the following two predicates:
\[
\operatorname{cover}	\df	(\forall v \in V)(\exists u \in V) \bigl( \operatorname{bin}(u) \land uv \notin E\bigr)
\]
\[
\operatorname{same}	\df	(\forall v,w \in V) \bigg( \left(\operatorname{notable}(v) \land \operatorname{notable}(w) \land vw \in E \right)
\]
\[
\quad\quad\quad\ \to (\forall u \in V : \operatorname{bin}(u))(uv \in E \Leftrightarrow uw \in E) \bigg)
\]

Now, we are ready to give the sentence $\varphi$ that describes the problem:
\[
\varphi= \psi \land \text{same} \land \text{cover}.
\]
This finishes the description of the reduction.
We are left with validating that the two instances are equivalent.

Suppose we were given a Yes-instance of \textsc{Unary $\ell$-Bin Packing} and let $\sigma \colon [n] \to [\ell]$ be the assignment of items to bins witnessing this fact.
The set $F$ contains an edge $uv_j$ for a special vertex $u \in C_i$ and a bin vertex $v_j$ whenever $\sigma(i) = j$.
Clearly, we have $|\{ e \in F \colon e \ni v_j \}| \le \sum_{i \in [n], \sigma(i) = j} s_i \le B = \tau$ for every vertex $v_j$ with $j \in [\ell]$, while every other vertex has at most one incident edge in $F$.
Now, we have to verify that $G \setminus F \models \varphi$.
We have $G \setminus F \models \operatorname{same} \land \operatorname{cover}$, since $\sigma$ is an assignment.
Finally, $G \setminus F \models \psi$, since our $F$ fulfills the condition of Claim~\ref{clm:psi_characterisation}.

Suppose now that we have a set $F$ of edges of $G$ with fair cost at most $B$ such that $G \setminus F \models \varphi$.
By Claim~\ref{clm:psi_characterisation} we have that $F$ contains only edges between special vertices in $G$ and bin vertices $v_1, \ldots, v_\ell$.
We partition the set of special vertices into sets $N$ and $R$; we put a special vertex $v$ in $N$ if $G \setminus F \models \operatorname{notable}(v)$ and we put it in $R$ otherwise.
Note that $|R| \le B$, since a vertex in $R$ contributes to the fair cost of every bin vertex $v_j$.
By Claim~\ref{clm:notable_properties} and the fact that $G \setminus F \models \operatorname{same}$, we have that a bin vertex $v_j$ is either completely attached or non-attached to $C_i \cap N$ for every $j\in[\ell]$ and every $i \in [n]$.
Furthermore, there are no edges between a vertex in $R$ and a bin vertex $v_j$ in $G \setminus F$ for every $j\in [\ell]$.
We define the assignment $\sigma \colon [n] \to [\ell]$ by defining $\sigma(i)$ to be the smallest integer such that there are no edges between $v_{\sigma(i)}$ and $C_i$ in $G \setminus F$.
This is well-defined: if $C_i\cap N\ne\emptyset$, use $\operatorname{cover}$ and $\operatorname{same}$.
Otherwise every bin is non-attached to $C_i$.
It remains to verify that the capacity condition is fulfilled.
For that, we have
\[
\sum_{i \in [n], \sigma(i) = j} s_i \le \left|\left\{e \in F \colon v_j \in e \right\}\right| \le \tau = B,
\]
since, for every $i$ with $\sigma(i)=j$, all $s_i$ edges between $v_j$ and the special vertices of $C_i$ belong to~$F$.
We conclude that $\sigma$ is a valid assignment.
Finally, deleting $v_1,\ldots,v_\ell$ leaves only the item cliques and isolated guards, so $\cvd(G)\le\ell$.
The construction has $3Bn+2\ell$ vertices, is polynomial in the unary input size, and uses a sentence whose length depends only on $\ell$, completing the reduction.
\end{proof}

\section{An XP Algorithm for Multiple Free Variables}\label{sec:lFairXP}

The algorithm processes the cliques outside a twin cover one at a time.
For each clique, it enumerates the number of vertices with each membership pattern in the solution sets.
The key point is that evaluating the formula requires only a bounded description of these choices, whereas the fair-cost constraints require exact neighborhood counts.
We combine these two kinds of information in a dynamic program, checking fairness inside each clique immediately and accumulating the contributions to the vertices of the twin cover.

\begin{proof}[Proof of Theorem~\ref{thm:lFairXP}]
Let $n=|V(G)|$ and $k=\tc(G)$.
Let $\tau$ be the fair-cost bound given in the input.
Compute a twin cover $K$ of size $k$ in \FPT time~\cite{Ganian15}, and guess all intersections $W_i\cap K$ for $i\in[\ell]$.
Put $r_0=\max\{1,2^{q_S}q_v\}$, $\xi=2^\ell r_0$, and $\beta=\alpha_\xi(q_S,q_v)$ (see Proposition~\ref{prop:lam} and Lemma~\ref{lem:induction}).

For each clique $C$ of $G-K$ with cover set $A$, enumerate the numbers $x_B$ of its vertices belonging to exactly the sets indexed by $B\subseteq[\ell]$, where $\sum_B x_B=|C|$.
Since the vertices of $C$ are interchangeable, these counts describe every choice within $C$ up to a permutation of its vertices.
There are at most $(|C|+1)^{2^\ell}$ such choices.
Write $s_i=\sum_{B\ni i}x_B$.
The fair-cost constraint for all vertices of $C$ and the set $W_i$ is equivalent to
\[
 |A\cap W_i\cap K|+s_i-[s_i=|C|]\leq\tau.
\]
Indeed, if $s_i<|C|$, a vertex of $C$ outside $W_i$ sees all $s_i$ selected vertices in $C$; otherwise, every vertex sees only $s_i-1$ of them.
We discard choices violating any of these inequalities.

The signature of a choice is $(A,(\min\{x_B,r_0\})_{B\subseteq[\ell]})$.
By Proposition~\ref{prop:lam}, each label class in a clique can be reduced to at most $r_0$ vertices while preserving the truth of $\varphi$.
The resulting cliques have size at most $\xi$, so Lemma~\ref{lem:induction} allows us to retain at most $\beta$ cliques of each signature.
Consequently, $G[K]$, the guessed labels on $K$, and the signature multiplicities truncated at $\beta$ determine the truth of $\varphi$.
They give a labeled representative of parameter-bounded size on which we can evaluate $\varphi$ by brute force.
We use this representative to test $\varphi$, while keeping the neighborhood contributions in the original graph exact for the fair-cost constraints.

We now process the cliques one by one by dynamic programming.
A state records these truncated multiplicities and, for every $a\in K$ and $i\in[\ell]$, the number $d_{a,i}$ of neighbors of $a$ in $W_i$ among the processed cliques; initially all counters are zero.
For a retained choice on a clique with cover set $A$, increment its signature multiplicity, truncating at $\beta$, and add $s_i$ to $d_{a,i}$ whenever $a\in A$.
Partial choices reaching the same state can be merged, since adding the same choices for the remaining cliques yields the same final state.

At the end, accept a state if
\[
 |N(a)\cap W_i\cap K|+d_{a,i}\leq\tau
 \qquad\text{for all }a\in K,\ i\in[\ell],
\]
and its labeled representative satisfies $\varphi$.
The transitions consider every possible assignment up to permutations within the cliques, and the local tests, final counters, and representative test verify precisely the required conditions.

The truncated multiplicities have only a parameter-dependent number of possibilities, while the $k\ell$ neighborhood counters each have at most $n+1$ values.
Each clique has at most $(n+1)^{2^\ell}$ choices.
Together with the evaluation of $\varphi$ on bounded representatives, this gives \XP running time, with a polynomial exponent depending only on $k$ and $\ell$.
\end{proof}

\section{Conclusions}\label{sec:conc}
\paragraph{Fair \Ll Edge Deletion problems.}

The crucial open problem is to resolve the parameterized complexity of the \textsc{Fair \FO Edge Deletion} problems for parameterization by neighborhood diversity and twin cover.
Observe that there is a big difference between vertex and edge deletion problems---in our hardness reduction, we use a deletion to an edgeless graph but a fair edge cost, in this case, equals the maximum degree of the former graph (and thus it is computable in polynomial time).

\paragraph{Generalization of parameters.}
Another open problem is to resolve the parameterized complexity of the \textsc{Fair \MSOo Vertex Evaluation} problems with respect to modular width and the formula size.

It seems that a more careful analysis of the model-checking algorithm may yield (again) sufficient insight into the structure of the fair solution and thus lead to an \FPT algorithm for this wider class of graphs.
However, it remains open whether this is possible and there is an \FPT algorithm for \textsc{Fair \MSOo Vertex Evaluation} for this parameterization or not.
The case of the parameterization by the cluster vertex deletion has been explored very recently~\cite{cvd} and \W{1}-hardness was provided already for \textsc{Fair \FO Vertex Deletion}, even when the formula size is included in the parameter, while \FPT algorithms were obtained for many natural fair vertex problems including \textsc{Fair VC}.

\paragraph{\MSO with Local Linear Constraints.}
Previously, an \FPT algorithm for the evaluation of a fair objective was given for the parameter neighborhood diversity~\cite{tamc}.
That algorithm was extended~\cite{KKMT} to so-called \emph{local linear constraints} again for a formula $\varphi(\cdot)$ with one free variable that is defined as follows.
Every vertex $v\in V(G)$ is accompanied by two nonnegative integers $\lambda(v), \mu(v)$, the lower and the upper bound, and the task is to find a set $X$ such that not only $G \models \varphi(X)$ but for each $v\in V(G)$ it holds that $\lambda(v) \le \left| N(v) \cap X \right| \le \mu(v)$.
Note that this is a generalization as a fair objective of value $\tau$ can be tested in this framework by setting $\lambda(v) = 0$ and $\mu(v) = \tau$ for every $v\in V(G)$.
Is this extension possible for parameterization by the twin cover number?

To support this question, we note that in the proof of Lemma~\ref{lem:computeX}, the minimal size of the neighborhood of the twin cover set $K$ for a shape in the exact neighborhood of $K$ is computed.
It is not hard to see that through a similar argument, we can compute the maximal size of the neighborhood of $K$ for a shape in the exact neighborhood of $K$.
Furthermore, lower and upper bounds for vertices in a clique can be assumed to be nearly the same -- each differs by at most 1~\cite{KKMT}.
Thus, Lemma~\ref{lem:computeX} gives only that if at least one of the computed bounds for a vertex $v$ in the twin cover is within $\lambda(v)$ and $\mu(v)$, then there is a realization satisfying the constraint at~$v$.
This does not guarantee a common realization satisfying the constraints at all vertices of the twin cover simultaneously.
Even for a single vertex, if both $\lambda(v), \mu(v)$ are in between the computed values, we do not know whether or not any of the desired values are attainable.
To see that not all values in the thus computed range are attainable, one can construct a formula that, for twin cliques up to a certain size, checks that the number of selected vertices is even.
Then, if the input graph contains only cliques up to this size, no twin cover vertex has an odd number of neighbors in the set $X$ (provided the cover vertices form an independent set).
Local linear constraints under the twin cover parameterization have been later explored by Knop~\cite{KnopLocal}.

\paragraph{\XP algorithms.}
Theorem~\ref{thm:lFairXP} complements the hardness in Theorem~\ref{thm:lFairVEisWH} with an \XP algorithm when the twin cover number, the number $\ell$ of free variables, and the quantifier length of the formula are parameters.
It would be interesting to obtain an \XP algorithm complementing \Cref{thm:cvdEdgeEval} as well, analogous to the algorithm parameterized by treewidth~\cite{Kolman09onfair}.

\paragraph{Towards new fair problems.}
As we proposed the examination of \textsc{Fair VC} already, we would like to turn our attention to exploring fair versions of other classical and well-studied \textsc{Vertex Deletion} problems. In contrast, certain \textsc{Fair Edge Deletion} problems have got some attention before, namely \textsc{Fair Feedback Edge Set}~\cite{LiSah} or \textsc{Fair Edge Odd Cycle Transversal}~\cite{KLS09}.
Besides \textsc{Fair VC}, we propose a study of \textsc{Fair Dominating Set} and \textsc{Fair Feedback Vertex Set}.
In particular, it would be really interesting to know whether fair variants of \textsc{Vertex Cover} and \textsc{Dominating Set} admit a similar behavior to that in the classical setting.
The \textsc{Fair Feedback Vertex Set} has been later studied by Kanesh, Maity, Muluk, and Saurabh~\cite{fairFVS}.

Furthermore, we would like to ask whether there is an \NP-hard Fair Vertex Deletion problem that admits an \FPT algorithm for parameterization by treedepth (and feedback vertex set) of the input graph.

\bigskip
\paragraph{Acknowledgements.} We thank an anonymous reviewer who suggested the current (simpler) version of Lemma~\ref{lem:computeX} and a problem to reduce from in \Cref{sec:hardness} that leads to a tighter ETH-based bound.
Our conference version used to prove Lemma~\ref{lem:computeX} with more complicated machinery of linear programming with objective minimization using a separable convex function. 
In \Cref{sec:hardness},  we reduced from the multicolor clique problem originally.
See previous versions of this article~\cite{KMTMFCS,arxiv-v4} for details on past implementation.

\bibliographystyle{alphaurl}%
\bibliography{fair}

\newcommand{\etalchar}[1]{$^{#1}$}
\begin{thebibliography}{CNNW14}

\bibitem[ALS91]{ALS:91}
Stefan Arnborg, Jens Lagergren, and Detlef Seese.
\newblock Easy problems for tree-decomposable graphs.
\newblock {\em Journal of Algorithms}, 12(2):308--340, 1991.
\newblock \href {https://doi.org/10.1016/0196-6774(91)90006-k}
  {\path{doi:10.1016/0196-6774(91)90006-k}}.

\bibitem[BDFH09]{bodlaenderDFH09}
Hans~L. Bodlaender, Rodney~G. Downey, Michael~R. Fellows, and Danny Hermelin.
\newblock On problems without polynomial kernels.
\newblock {\em Journal of Computer and System Sciences}, 75(8):423--434, 2009.
\newblock \href {https://doi.org/10.1016/j.jcss.2009.04.001}
  {\path{doi:10.1016/j.jcss.2009.04.001}}.

\bibitem[BJG09]{digraphs-book}
J{\o}rgen Bang-Jensen and Gregory~Z. Gutin.
\newblock {\em Digraphs: Theory, Algorithms and Applications}.
\newblock Springer Monographs in Mathematics. Springer, London, second edition,
  2009.
\newblock \href {https://doi.org/10.1007/978-1-84800-998-1}
  {\path{doi:10.1007/978-1-84800-998-1}}.

\bibitem[BLM20]{Belmonte}
R{\'e}my Belmonte, Michael Lampis, and Valia Mitsou.
\newblock Parameterized (approximate) defective coloring.
\newblock {\em SIAM Journal on Discrete Mathematics}, 34(2):1084--1106, 2020.
\newblock \href {https://doi.org/10.1137/18m1223666}
  {\path{doi:10.1137/18m1223666}}.

\bibitem[CCW86]{defcol}
Lenore~J. Cowen, Robert Cowen, and Douglas~R. Woodall.
\newblock Defective colorings of graphs in surfaces: Partitions into subgraphs
  of bounded valency.
\newblock {\em Journal of Graph Theory}, 10(2):187--195, 1986.
\newblock \href {https://doi.org/10.1002/jgt.3190100207}
  {\path{doi:10.1002/jgt.3190100207}}.

\bibitem[CFK{\etalchar{+}}15]{CyganFKLMPPS15}
Marek Cygan, Fedor~V. Fomin, {\L}ukasz Kowalik, Daniel Lokshtanov, D{\'a}niel
  Marx, Marcin Pilipczuk, Micha{\l} Pilipczuk, and Saket Saurabh.
\newblock {\em Parameterized Algorithms}.
\newblock Springer, 2015.
\newblock \href {https://doi.org/10.1007/978-3-319-21275-3}
  {\path{doi:10.1007/978-3-319-21275-3}}.

\bibitem[CHPT24]{TCHP08-a}
Derek Corneil, Michel Habib, Christophe Paul, and Marc Tedder.
\newblock A recursive linear time modular decomposition algorithm via {LexBFS},
  2024.
\newblock \href {https://arxiv.org/abs/0710.3901v4} {\path{arXiv:0710.3901v4}}.

\bibitem[CM93]{CourcelleMosbah}
Bruno Courcelle and Mohamed Mosbah.
\newblock Monadic second-order evaluations on tree-decomposable graphs.
\newblock {\em Theoretical Computer Science}, 109(1--2):49--82, 1993.
\newblock \href {https://doi.org/10.1016/0304-3975(93)90064-z}
  {\path{doi:10.1016/0304-3975(93)90064-z}}.

\bibitem[CNNW14]{ChopinNNW14}
Morgan Chopin, Andr{\'{e}} Nichterlein, Rolf Niedermeier, and Mathias Weller.
\newblock Constant thresholds can make target set selection tractable.
\newblock {\em Theory of Computing Systems}, 55(1):61--83, 2014.
\newblock \href {https://doi.org/10.1007/s00224-013-9499-3}
  {\path{doi:10.1007/s00224-013-9499-3}}.

\bibitem[Cou90]{cour}
Bruno Courcelle.
\newblock The monadic second-order logic of graphs. {I.} recognizable sets of
  finite graphs.
\newblock {\em Information and Computation}, 85(1):12--75, 1990.
\newblock \href {https://doi.org/10.1016/0890-5401(90)90043-h}
  {\path{doi:10.1016/0890-5401(90)90043-h}}.

\bibitem[DF13]{df13}
Rodney~G. Downey and Michael~R. Fellows.
\newblock {\em Fundamentals of Parameterized Complexity}.
\newblock Texts in Computer Science. Springer, 2013.
\newblock \href {https://doi.org/10.1007/978-1-4471-5559-1}
  {\path{doi:10.1007/978-1-4471-5559-1}}.

\bibitem[Dru15]{Drucker15}
Andrew Drucker.
\newblock New limits to classical and quantum instance compression.
\newblock {\em SIAM Journal on Computing}, 44(5):1443--1479, 2015.
\newblock \href {https://doi.org/10.1137/130927115}
  {\path{doi:10.1137/130927115}}.

\bibitem[FG04]{frickG04}
Markus Frick and Martin Grohe.
\newblock The complexity of first-order and monadic second-order logic
  revisited.
\newblock {\em Annals of Pure and Applied Logic}, 130(1):3--31, 2004.
\newblock Papers presented at the 2002 IEEE Symposium on Logic in Computer
  Science (LICS).
\newblock \href {https://doi.org/10.1016/j.apal.2004.01.007}
  {\path{doi:10.1016/j.apal.2004.01.007}}.

\bibitem[FLSZ19]{Kernelization}
Fedor~V. Fomin, Daniel Lokshtanov, Saket Saurabh, and Meirav Zehavi.
\newblock {\em Kernelization: Theory of Parameterized Preprocessing}.
\newblock Cambridge University Press, Cambridge, 2019.
\newblock \href {https://doi.org/10.1017/9781107415157}
  {\path{doi:10.1017/9781107415157}}.

\bibitem[Gan12]{iwpec/Ganian11}
Robert Ganian.
\newblock Twin-cover: Beyond vertex cover in parameterized algorithmics.
\newblock In D{\'{a}}niel Marx and Peter Rossmanith, editors, {\em
  Parameterized and Exact Computation: 6th International Symposium ({IPEC}
  2011)}, volume 7112 of {\em Lecture Notes in Computer Science}, pages
  259--271. Springer, 2012.
\newblock \href {https://doi.org/10.1007/978-3-642-28050-4_21}
  {\path{doi:10.1007/978-3-642-28050-4_21}}.

\bibitem[Gan15]{Ganian15}
Robert Ganian.
\newblock Improving vertex cover as a graph parameter.
\newblock {\em Discrete Mathematics {\&} Theoretical Computer Science},
  17(2):77--100, 2015.
\newblock \href {https://doi.org/10.46298/dmtcs.2136}
  {\path{doi:10.46298/dmtcs.2136}}.

\bibitem[GH15]{gajarskyH15}
Jakub Gajarsk{\'{y}} and Petr Hlin{\v{e}}n{\'{y}}.
\newblock Kernelizing {MSO} properties of trees of fixed height, and some
  consequences.
\newblock {\em {Logical Methods in Computer Science}}, 11(1):1--26, 2015.
\newblock Article 19.
\newblock \href {https://doi.org/10.2168/LMCS-11(1:19)2015}
  {\path{doi:10.2168/LMCS-11(1:19)2015}}.

\bibitem[GHN{\etalchar{+}}19]{GanianHNOM19}
Robert Ganian, Petr Hlin{\v{e}}n{\'{y}}, Jaroslav Ne{\v{s}}et{\v{r}}il, Jan
  Obdr{\v{z}}{\'{a}}lek, and Patrice {Ossona de Mendez}.
\newblock Shrub-depth: Capturing height of dense graphs.
\newblock {\em Logical Methods in Computer Science}, 15(1):7:1--7:25, 2019.
\newblock \href {https://doi.org/10.23638/lmcs-15(1:7)2019}
  {\path{doi:10.23638/lmcs-15(1:7)2019}}.

\bibitem[GHO{\etalchar{+}}20]{GajarskyHOLR16}
Jakub Gajarsk{\'y}, Petr Hlin{\v e}n{\'y}, Jan Obdr{\v z}{\'a}lek, Daniel
  Lokshtanov, and Maadapuzhi~Sridharan Ramanujan.
\newblock A new perspective on {FO} model checking of dense graph classes.
\newblock {\em {ACM} Transactions on Computational Logic}, 21(4), July 2020.
\newblock \href {https://doi.org/10.1145/3383206} {\path{doi:10.1145/3383206}}.

\bibitem[GKS17]{GroheKS17}
Martin Grohe, Stephan Kreutzer, and Sebastian Siebertz.
\newblock Deciding first-order properties of nowhere dense graphs.
\newblock {\em Journal of the ACM}, 64(3):17:1--17:32, 2017.
\newblock \href {https://doi.org/10.1145/3051095} {\path{doi:10.1145/3051095}}.

\bibitem[GLO13]{GLO13}
Jakub Gajarsk{\'{y}}, Michael Lampis, and Sebastian Ordyniak.
\newblock Parameterized algorithms for modular-width.
\newblock In {\em Parameterized and Exact Computation: 8th International
  Symposium ({IPEC} 2013)}, volume 8246 of {\em Lecture Notes in Computer
  Science}, pages 163--176. Springer, 2013.
\newblock \href {https://doi.org/10.1007/978-3-319-03898-8_15}
  {\path{doi:10.1007/978-3-319-03898-8_15}}.

\bibitem[GO13]{GO:13}
Robert Ganian and Jan Obdr{\v{z}}{\'{a}}lek.
\newblock Expanding the expressive power of monadic second-order logic on
  restricted graph classes.
\newblock In Thierry Lecroq and Laurent Mouchard, editors, {\em Combinatorial
  Algorithms: 24th International Workshop ({IWOCA} 2013)}, volume 8288 of {\em
  Lecture Notes in Computer Science}, pages 164--177. Springer, 2013.
\newblock \href {https://doi.org/10.1007/978-3-642-45278-9_15}
  {\path{doi:10.1007/978-3-642-45278-9_15}}.

\bibitem[HKS09]{defectiPractical1}
Fr{\'e}d{\'e}ric Havet, Ross~J. Kang, and Jean-S{\'e}bastien Sereni.
\newblock Improper coloring of unit disk graphs.
\newblock {\em Networks}, 54(3):150--164, 2009.
\newblock \href {https://doi.org/10.1002/net.20318}
  {\path{doi:10.1002/net.20318}}.

\bibitem[JKMS13]{JansenKMS13}
Klaus Jansen, Stefan Kratsch, D{\'{a}}niel Marx, and Ildik{\'{o}} Schlotter.
\newblock Bin packing with fixed number of bins revisited.
\newblock {\em Journal of Computer and System Sciences}, 79(1):39--49, 2013.
\newblock \href {https://doi.org/10.1016/j.jcss.2012.04.004}
  {\path{doi:10.1016/j.jcss.2012.04.004}}.

\bibitem[KD79]{KriDeo}
Mukkai~S. Krishnamoorthy and Narsingh Deo.
\newblock Node-deletion {NP}-complete problems.
\newblock {\em SIAM Journal on Computing}, 8(4):619--625, 1979.
\newblock \href {https://doi.org/10.1137/0208049} {\path{doi:10.1137/0208049}}.

\bibitem[KKMT19]{KKMT}
Du{\v s}an Knop, Martin Kouteck{\'y}, Tom{\'a}{\v s} Masa{\v r}{\'i}k, and
  Tom{\'a}{\v s} Toufar.
\newblock Simplified algorithmic metatheorems beyond {MSO}: Treewidth and
  neighborhood diversity.
\newblock {\em Logical Methods in Computer Science}, 15(4):12:1--12:32, 2019.
\newblock \href {https://doi.org/10.23638/LMCS-15(4:12)2019}
  {\path{doi:10.23638/LMCS-15(4:12)2019}}.

\bibitem[Kle67]{Klein67}
Morton Klein.
\newblock A primal method for minimal cost flows with applications to the
  assignment and transportation problems.
\newblock {\em Management Science}, 14(3):205--220, 1967.
\newblock \href {https://doi.org/10.1287/mnsc.14.3.205}
  {\path{doi:10.1287/mnsc.14.3.205}}.

\bibitem[KLS09]{KLS09}
Petr Kolman, Bernard Lidick{\'y}, and Jean-S{\'e}bastien Sereni.
\newblock On fair edge deletion problems, 2009.
\newblock Manuscript.
\newblock URL: \url{https://kam.mff.cuni.cz/~kolman/papers/kls09.pdf}.

\bibitem[KLS10]{Kolman09onfair}
Petr Kolman, Bernard Lidick{\'y}, and Jean-S{\'e}bastien Sereni.
\newblock Fair edge deletion problems on tree-decomposable graphs and improper
  colorings, 2010.
\newblock Manuscript.
\newblock URL: \url{https://lidicky.name/pub/kls10.pdf}.

\bibitem[KMMS21]{fairFVS}
Lawqueen Kanesh, Soumen Maity, Komal Muluk, and Saket Saurabh.
\newblock Parameterized complexity of fair feedback vertex set problem.
\newblock {\em Theoretical Computer Science}, 867:1--12, 2021.
\newblock \href {https://doi.org/10.1016/j.tcs.2021.03.008}
  {\path{doi:10.1016/j.tcs.2021.03.008}}.

\bibitem[KMS08]{defectiPractical2}
Ross~J. Kang, Tobias M{\"u}ller, and Jean-S{\'e}bastien Sereni.
\newblock Improper colouring of (random) unit disk graphs.
\newblock {\em Discrete Mathematics}, 308(8):1438--1454, 2008.
\newblock Third European Conference on Combinatorics.
\newblock \href {https://doi.org/10.1016/j.disc.2007.07.070}
  {\path{doi:10.1016/j.disc.2007.07.070}}.

\bibitem[KMT19]{KMTMFCS}
Du{\v s}an Knop, Tom{\'a}{\v s} Masa{\v r}{\'i}k, and Tom{\'a}{\v s} Toufar.
\newblock Parameterized complexity of fair vertex evaluation problems.
\newblock In Peter Rossmanith, Pinar Heggernes, and Joost-Pieter Katoen,
  editors, {\em 44th International Symposium on Mathematical Foundations of
  Computer Science ({MFCS} 2019)}, volume 138 of {\em Leibniz International
  Proceedings in Informatics}, pages 33:1--33:16, Dagstuhl, Germany, 2019.
  Schloss Dagstuhl--Leibniz-Zentrum f{\"u}r Informatik.
\newblock \href {https://doi.org/10.4230/LIPIcs.MFCS.2019.33}
  {\path{doi:10.4230/LIPIcs.MFCS.2019.33}}.

\bibitem[KMT20]{arxiv-v4}
Du{\v s}an Knop, Tom{\'a}{\v s} Masa{\v r}{\'i}k, and Tom{\'a}{\v s} Toufar.
\newblock Parameterized complexity of fair vertex evaluation problems, 2020.
\newblock \href {https://arxiv.org/abs/1803.06878v4}
  {\path{arXiv:1803.06878v4}}.

\bibitem[KN11]{kontienN11}
Juha Kontinen and Hannu Niemist{\"{o}}.
\newblock Extensions of {MSO} and the monadic counting hierarchy.
\newblock {\em Information and Computation}, 209(1):1--19, 2011.
\newblock \href {https://doi.org/10.1016/j.ic.2010.09.002}
  {\path{doi:10.1016/j.ic.2010.09.002}}.

\bibitem[Kno21]{KnopLocal}
Du{\v s}an Knop.
\newblock Local linear set on graphs with bounded twin cover number.
\newblock {\em Information Processing Letters}, 170:106118, 2021.
\newblock \href {https://doi.org/10.1016/j.ipl.2021.106118}
  {\path{doi:10.1016/j.ipl.2021.106118}}.

\bibitem[Lam12]{lam}
Michael Lampis.
\newblock Algorithmic meta-theorems for restrictions of treewidth.
\newblock {\em Algorithmica}, 64(1):19--37, 2012.
\newblock \href {https://doi.org/10.1007/s00453-011-9554-x}
  {\path{doi:10.1007/s00453-011-9554-x}}.

\bibitem[Lam14]{Lampis:13}
Michael Lampis.
\newblock Model checking lower bounds for simple graphs.
\newblock {\em Logical Methods in Computer Science}, 10(1):1--21, 2014.
\newblock Article 18.
\newblock \href {https://doi.org/10.2168/LMCS-10(1:18)2014}
  {\path{doi:10.2168/LMCS-10(1:18)2014}}.

\bibitem[Lib04]{LibkinFMT}
Leonid Libkin.
\newblock {\em Elements of Finite Model Theory}.
\newblock Texts in Theoretical Computer Science. An {EATCS} Series. Springer,
  2004.
\newblock \href {https://doi.org/10.1007/978-3-662-07003-1}
  {\path{doi:10.1007/978-3-662-07003-1}}.

\bibitem[LS89]{LiSah}
Li-Shin Lin and Sartaj Sahni.
\newblock Fair edge deletion problems.
\newblock {\em IEEE Transactions on Computers}, 38(5):756--761, 1989.
\newblock \href {https://doi.org/10.1109/12.24280}
  {\path{doi:10.1109/12.24280}}.

\bibitem[Mar10]{MarxPSI}
D{\'a}niel Marx.
\newblock Can you beat treewidth?
\newblock {\em Theory of Computing}, 6:85--112, 2010.
\newblock \href {https://doi.org/10.4086/toc.2010.v006a005}
  {\path{doi:10.4086/toc.2010.v006a005}}.

\bibitem[MN09]{kapitoly}
Ji{\v{r}}{\'{\i}} Matou{\v{s}}ek and Jaroslav Ne{\v{s}}et{\v{r}}il.
\newblock {\em Invitation to Discrete Mathematics}.
\newblock Oxford University Press, second edition, 2009.

\bibitem[MOZ25]{cvd}
Tom{\'a}{\v s} Masa{\v r}{\'i}k, J{\k e}drzej Olkowski, and Anna
  {Zych-Pawlewicz}.
\newblock Fair vertex problems parameterized by cluster vertex deletion, 2025.
\newblock \href {https://arxiv.org/abs/2502.01400} {\path{arXiv:2502.01400}}.

\bibitem[MT20]{tamc}
Tom{\'{a}}{\v{s}} Masa{\v{r}}{\'{\i}}k and Tom{\'{a}}{\v{s}} Toufar.
\newblock Parameterized complexity of fair deletion problems.
\newblock {\em Discrete Applied Mathematics}, 278:51--61, 2020.
\newblock \href {https://doi.org/10.1016/j.dam.2019.06.001}
  {\path{doi:10.1016/j.dam.2019.06.001}}.

\bibitem[NO12]{sparsity}
Jaroslav Ne{\v s}et{\v r}il and Patrice {Ossona de Mendez}.
\newblock {\em Sparsity: Graphs, Structures, and Algorithms}, volume~28 of {\em
  Algorithms and Combinatorics}.
\newblock Springer, 2012.
\newblock \href {https://doi.org/10.1007/978-3-642-27875-4}
  {\path{doi:10.1007/978-3-642-27875-4}}.

\bibitem[OWW14]{oertelWW14}
Timm Oertel, Christian Wagner, and Robert Weismantel.
\newblock Integer convex minimization by mixed integer linear optimization.
\newblock {\em Operations Research Letters}, 42(6--7):424--428, 2014.
\newblock \href {https://doi.org/10.1016/j.orl.2014.07.005}
  {\path{doi:10.1016/j.orl.2014.07.005}}.

\bibitem[Pil11]{Pilipczuk11}
Micha\l{} Pilipczuk.
\newblock Problems parameterized by treewidth tractable in single exponential
  time: {A} logical approach.
\newblock In Filip Murlak and Piotr Sankowski, editors, {\em Mathematical
  Foundations of Computer Science: 36th International Symposium ({MFCS} 2011)},
  volume 6907 of {\em Lecture Notes in Computer Science}, pages 520--531.
  Springer, 2011.
\newblock \href {https://doi.org/10.1007/978-3-642-22993-0_47}
  {\path{doi:10.1007/978-3-642-22993-0_47}}.

\bibitem[See96]{Seese96}
Detlef Seese.
\newblock Linear time computable problems and first-order descriptions.
\newblock {\em Mathematical Structures in Computer Science}, 6(6):505--526,
  1996.
\newblock \href {https://doi.org/10.1017/s0960129500070079}
  {\path{doi:10.1017/s0960129500070079}}.

\bibitem[TCHP08]{TCHP08}
Marc Tedder, Derek~G. Corneil, Michel Habib, and Christophe Paul.
\newblock Simpler linear-time modular decomposition via recursive factorizing
  permutations.
\newblock In {\em Automata, Languages and Programming: 35th International
  Colloquium ({ICALP} 2008), Part I}, volume 5125 of {\em Lecture Notes in
  Computer Science}, pages 634--645. Springer, 2008.
\newblock \href {https://doi.org/10.1007/978-3-540-70575-8_52}
  {\path{doi:10.1007/978-3-540-70575-8_52}}.

\bibitem[WAN83]{Watanabe}
Toshimasa Watanabe, Tadashi Ae, and Akira Nakamura.
\newblock On the {NP}-hardness of edge-deletion and -contraction problems.
\newblock {\em Discrete Applied Mathematics}, 6(1):63--78, 1983.
\newblock \href {https://doi.org/10.1016/0166-218X(83)90101-4}
  {\path{doi:10.1016/0166-218X(83)90101-4}}.

\bibitem[Yan78]{Yannakakis78}
Mihalis Yannakakis.
\newblock Node- and edge-deletion {NP}-complete problems.
\newblock In {\em Proceedings of the 10th Annual ACM Symposium on Theory of
  Computing ({STOC} 1978)}, pages 253--264. Association for Computing
  Machinery, 1978.
\newblock \href {https://doi.org/10.1145/800133.804355}
  {\path{doi:10.1145/800133.804355}}.

\bibitem[Yan81]{Yannakakis81}
Mihalis Yannakakis.
\newblock Edge-deletion problems.
\newblock {\em SIAM Journal on Computing}, 10(2):297--309, 1981.
\newblock \href {https://doi.org/10.1137/0210021} {\path{doi:10.1137/0210021}}.

\end{thebibliography}

\end{document}